\documentclass[journal]{new-aiaa}
\usepackage[utf8]{inputenc}

\usepackage{xcolor}
\usepackage{graphicx}
\graphicspath{{Figures/}} %Setting the graphicspath
\usepackage{mathtools}
\usepackage{amsmath}
\usepackage{amsfonts}

\usepackage{amssymb}

\usepackage{amsthm}
\usepackage{booktabs}
\usepackage{soul}
\usepackage{enumitem}
\usepackage{bm}
\usepackage{float}
\usepackage{cleveref}
\usepackage{natbib}
\usepackage{setspace}
\usepackage[justification=centering]{caption}
\makeatletter
\g@addto@macro\normalsize{%
  \setlength\abovedisplayskip{8pt}
  \setlength\belowdisplayskip{8pt}
  \setlength\abovedisplayshortskip{8pt}
  \setlength\belowdisplayshortskip{8pt}
}
\makeatother

\newcommand{\norm}[1]{{\left\lVert#1\right\rVert}}
\DeclarePairedDelimiter\abs{\lvert}{\rvert}%

\def\laplace#1{\mathfrak{L}[#1]}

\def\lone{{\mathcal{L}_1}}

\def\lonew{${\mathcal{L}_1}$ }

\def\reft{{{\textup{ref}}}}

\def \idt{{\textup{id}}}
\def\laplace#1{\mathfrak{L}[#1]}

\def\laplace#1{\mathfrak{L}[#1]}

\def\mbR{\mathbb{R}}
\def\mbI{\mathbb{I}}

\def\mcC{\mathcal{C}}

\def\mcF{\mathcal{F}}

\def\trieq{\triangleq}

\def\l1rho{L_{1\rho}}
\def\l2rho{L_{2\rho}}

\def\b10{b_{10}}
\def\b20{b_{20}}
\def\bi0{b_{i0}}
\makeatletter
\def\cl@part {\@elt {chapter}}
\makeatother
\usepackage{cleveref}

\crefname{equation}{}{} %skip "eq" or "eqs". 
\crefname{lem}{Lemma}{Lemmas}
\crefname{thm}{Theorem}{Theorems}
\crefname{table}{Table}{Tables}
\crefname{figure}{Fig.}{Figs.}
\crefname{rem}{Remark}{Remarks}
\crefname{assum}{Assumption}{Assumptions}
\crefname{section}{Section}{Sections}
\crefname{defn}{Definition}{Definitions}

\hypersetup{
    colorlinks=true,
    pdfborderstyle={/S/U/W 1},
    linkcolor=blue,
    filecolor=magenta,      
    urlcolor= cyan,
    linkbordercolor=red,% hyperlink border will be red
    citecolor  = cyan
}

\newtheorem{theorem}{Theorem}
\newtheorem{lemma}{Lemma}

\theoremstyle{defn}  
\theoremstyle{definition} 
\newtheorem{assumption}{Assumption}
\theoremstyle{remark}  \newtheorem{remark}{Remark}

\newcommand{\II}{\mathbb{I}}
\newcommand{\IR}{\mathbb{R}}

\newcommand{\sat}[3]{\text{Sat}_{#1}^{#2} \left( {#3} \right)}

\DeclareMathOperator*{\trace}{tr}

\def \loneAC {$\lone$AC}

\usepackage{url}

    \makeatletter
\def\@fnsymbol#1{\ensuremath{\ifcase#1\or *\or \ddagger\or
   \mathsection\or \mathparagraph\or \|\or **\or \dagger\dagger
   \or \ddagger\ddagger \else\@ctrerr\fi}}
    \makeatother

\begin{document}
\title{\lonew Adaptive Control with Switched Reference Models: Application to Learn-to-Fly}
\vspace{-.9cm}

\author{Steven Snyder \footnote{Research Engineer, Dynamics Systems and Control Branch; \url{steven.m.snyder@nasa.gov}. Member AIAA.}}
\affil{NASA Langley Research Center, Hampton, VA 23681
}
\author{Pan Zhao\footnote{Postdoctoral Researcher, Department of Mechanical Science and Engineering; \url{panzhao2@illinois.edu}. Corresponding author.} and Naira Hovakimyan \footnote{Professor, Department of Mechanical Science and Engineering; \url{nhovakim@illinois.edu}. Fellow AIAA.}}
\affil{University of Illinois at Urbana-Champaign, Urbana,
Illinois 61801}

\maketitle
\begin{abstract}
Learn-to-Fly (L2F) is a new framework that aims to replace the traditional iterative  development paradigm for aerial vehicles with a combination of real-time aerodynamic modeling, guidance, and learning control. To ensure safe learning of the vehicle dynamics on the fly, this paper presents an \lonew adaptive control (\loneAC) based scheme, which actively estimates and compensates for the discrepancy between the intermediately learned dynamics and the actual dynamics.  First, to incorporate the periodic update of the learned model within the L2F framework, this paper extends the \loneAC ~architecture to handle a switched reference system subject to unknown time-varying parameters and disturbances. The paper also includes analysis of both transient and steady-state performance of the \loneAC ~architecture in the presence of non-zero initialization error for the state predictor.
Second,  the paper presents how the proposed \loneAC ~scheme is integrated into the L2F framework, including its interaction with the baseline controller and the real-time modeling module. Finally, flight tests on an unmanned aerial vehicle (UAV) validate the efficacy of the proposed control and learning scheme.  
\end{abstract}

% \begin{keyword}
% Adaptive control; LPV system; Learn-to-fly; Safety-critical system; Flight control
% \end{keyword}

\section{Introduction}\label{sec:introduction}
Historically, aerodynamic modeling from computational fluid dynamics (CFD) data, wind tunnel test data, or flight test data has been a multi-step process where the tests are performed, followed by post-test analysis of the data.  
During this data analysis, it is often discovered that the data content has insufficient richness in particular areas, requiring additional testing to fill in the data gaps.  
This results in a time-consuming process that may involve a large number of test sorties to acquire the necessary data for development of an adequate aerodynamic model. 
In view of these disadvantages, there has been an effort to develop the so-called {\it Learn-to-Fly} (L2F) framework, which aims to replace the current iterative vehicle development paradigm with a combination of real-time aerodynamic modeling, learning control, and other 
techniques, enabling real-time, self-learning flight vehicles \cite{Heim2018L2F-overview}.  
The advantages of inflight modeling include the ability to analyze model fidelity and data richness while gathering data, which can in turn reduce the number of test flight sorties required.
The conventional aircraft development process and the L2F concept are depicted in Figure \ref{fig:convVsL2F}. 
Within the L2F framework, state-of-the-art in-flight aerodynamic modeling approaches are combined with guidance and learning control law design methods on board the aircraft, running in real time.  
However, during the initial learning phase, the only available model of the vehicle's aerodynamics is based on a guess, which may be very inaccurate, depending on the amount of ground testing performed.  
The control law must be able to ensure stability of the vehicle with this potentially poor model.
As the vehicle flies and additional data are collected, the aerodynamic model is improved and, consequently, the controller should adjust to provide a better closed-loop response, e.g., improved tracking performance. 
In the subsequent narration, we use {real-time modeling}, {model learning} and {system identification} interchangeably. 
\begin{figure}[htbp]
    \centering
    \includegraphics[width=1.0\linewidth]{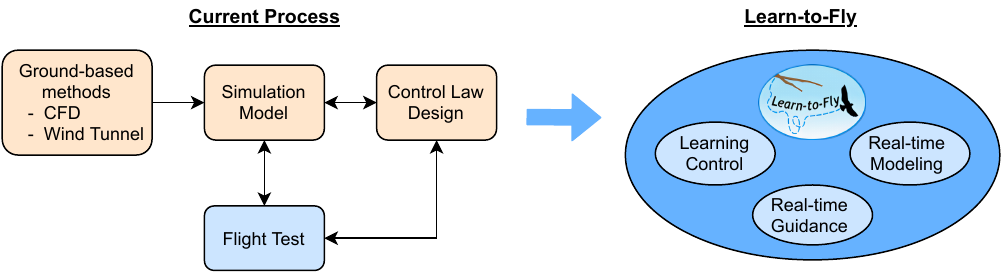}
    \caption{Conventional aircraft development process vs. Learn-to-Fly concept (reproduction of Fig.~1 in \cite{Heim2018L2F-overview})}
    \label{fig:convVsL2F}
\end{figure}

Since the available dynamic model of the vehicle could be quite poor at the initial stage of model learning, the on-board controller should be capable of stabilizing the vehicle despite the presence of large uncertainties and disturbances. 
Among different control techniques, {adaptive control} is one of the most promising methods to handle large uncertainties while providing good nominal performance. 
In the last several decades, the field of adaptive control has witnessed tremendous developments, capturing different classes of linear and nonlinear systems, subject to various parametric and state-dependent uncertainties, and unmodeled dynamics \citep{KKK-book,ioannou2012robust}. 
As one of the most well-known adaptive control techniques, {model reference adaptive control (MRAC)} deserves special attention. 
However, as experience has shown, conventional MRAC controllers lose robustness guarantees in the presence of fast adaptation, and therefore are only effective for slowly-varying uncertainties. 
\lonew adaptive control (\loneAC) resolved this issue with an appropriate filtering structure that decouples the estimation loop from the control loop \cite{naira2010l1book,Cao2008TAC}. 
As a result, arbitrarily fast adaptation, subject only to hardware limitations, can be used without sacrificing robustness. 
Note that fast adaptation is desirable for timely estimation of and compensation for the (fast-varying) uncertainties/disturbances.   
Benefiting from the fast and robust adaptation, \loneAC ~provides {guaranteed tracking performance}, both in {transient} and steady-state, and robustness (e.g., nonzero time delay margin), in the presence of various  uncertainties and disturbances \cite{naira2010l1book}.  
We note that \loneAC ~has been verified on a large number of real-world systems \citep{naira2010l1book,naira2011L1_Safety}, including NASA's subscale commercial jet \citep{gregory2009l1,gregory2010flight}, and Calspan's Learjet (a piloted aircraft) \citep{ackerman2017evaluation,puig2019learjet}.

When incorporating the adaptive control scheme into the L2F framework, we need to address a crucial issue, i.e., {updating the reference model} as the data collection and model learning processes evolve. 
For instance, at some point, assuming the real-time modeling module produces a model, $M_1$, we may use $M_1$ to produce a reference model to design an adaptive controller that tries to cancel the discrepancy between the actual uncertain dynamics and $M_1$. 
After a certain amount of time, assuming the model is updated to be $M_2$, a natural choice is to use this new model, $M_2$ to replace the old model, $M_1$, when generating the desired/nominal model for adaptive control design. 
One might wish to change the specified desired dynamics based on the updated model to ensure model matching conditions are satisfied or to choose a more natural set of desired dynamics from the potentially infinite number that meet the given design criteria (such as those in miltary or industry standards).
Locking in a specific set of desired dynamics when little is known about the vehicle takes away the option of many other viable designs.
Essentially, implementation of the L2F framework leads to a {switched reference model} for an adaptive control scheme. 
Therefore, we need to address the stability and performance of a switching adaptive control scheme.  

There have been only a few attempts to address adaptive control design for a switched reference model.  
Sang and Tao studied MRAC design for a piecewise linear (PWL) system with parametric uncertainties \cite{sang2011adaptive-pwl}, in which asymptotic tracking performance is achieved under arbitrarily fast switching when a common Lyapunov function exists for the PWL reference system. 
In the absence of such common Lyapunov function, only small tracking error in the mean-squared sense can be achieved under slow switching conditions. 
Xie and Zhao \cite{xie2016mrac-switching-lpv} studied MRAC design for a switched linear parameter-varying (LPV) system subject to parametric uncertainties under average-dwell-time switching. 
Additionally, Wu et al. \cite{wu2015adaptive-switch-async} investigated the issues of asynchronous switching in MRAC design for a switched linear system with parametric uncertainties. 
All these methods are based on a {direct} MRAC scheme, and only {asymptotic} tracking performance can be achieved at most. 
On the contrary, the \loneAC ~scheme proposed in this paper and also in our recent work \cite{snyder2019switchingL1}, to the best of our knowledge, is the first {indirect} adaptive control scheme considering a switched reference model. 
Additionally, the \loneAC ~scheme allows for consideration of time-varying parametric uncertainties and (time-dependent) disturbances instead of just the time-invariant parametric uncertainties considered in \cite{sang2011adaptive-pwl,xie2016mrac-switching-lpv,snyder2019switchingL1}. 
Finally, with the help of the \loneAC ~architecture, we will provide guaranteed stability and {transient} tracking performance in the presence of a switched reference model.

{\bf Statement of Contributions}. Motivated by the L2F concept, this paper presents an \loneAC ~based scheme to ensure safe learning of unknown dynamics of an aerial vehicle.
The proposed \loneAC ~is intended to provide the necessary robustness required to keep the vehicle flying while a learning module (e.g., \cite{L2FModeling}) learns the dynamics.
At the core is an \loneAC ~architecture with a piecewise constant (PWC) adaptive law  for a switched linear system to characterize the change of the desired/nominal dynamics as a result of periodic model update from the real-time system identification module. 
For stability and performance analysis, we introduce a {virtual} reference system that depends on the true uncertainties and represents the best achievable performance. 
Under reasonable assumptions,  we  derive  transient  performance  bounds in terms of the inputs and states of the adaptive system as compared to the same signals of the reference system. 
The influence of re-initialization error associated with the state predictor, a key component of an \lonew controller, is also analyzed. 
We then illustrate how the proposed \loneAC ~scheme is incorporated into the L2F framework. 
We finally present the results of flight tests of the proposed control and learning schemes on a fixed-wing unmanned aerial vehicle (UAV) conducted by NASA Langley Research Center (LaRC) in Virginia, USA. 
To the best of our knowledge, these flight tests were the first of this vehicle class enabled by the L2F concept, i.e., reliance on a combination of real-time modeling and robust adaptive control while minimizing time-consuming and expensive ground testing. 
To summarize, our contributions include
\begin{enumerate}[label={(\arabic*)},noitemsep,topsep=0pt]
    \item an \loneAC ~architecture with a PWC adaptive law for a switched linear system subject to time-varying parametric uncertainties and disturbances, with guaranteed robustness and transient performance;
    \item an illustration of how the proposed \loneAC ~architecture is incorporated into the control law design and update within L2F framework; and
    \item validation of the proposed \loneAC ~in the L2F framework with an existing learning scheme \cite{L2FModeling} using flight tests on a fixed-wing UAV. 
\end{enumerate}
In comparison, our preliminary work \cite{snyder2019switchingL1} addresses only time-invariant parametric uncertainties, uses a projection-based adaptive law that is not as amenable to implementation as the PWC law adopted in this paper (more details given in the first paragraph of \cref{sec::L1AC}), is not relevant to L2F, and is only validated on a numerical example. 
The current work uses L2F as motivation for its development of \loneAC ~with switched desired dynamics, but is also applicable in other situations where it may be advantageous for the desired dynamics to switch, such as having different dynamics for different tasks (e.g., precision control takeoff/landing vs. energy-efficient cruise/loiter) or when vehicle mass properties change (e.g., package pick up and delivery).
We do not seek to debate the benefits of the L2F paradigm.

The paper is organized as follows.  Section~\ref{sec:l2f-overview} gives an overview of the L2F framework. Section~\ref{sec::L1AC} presents an \loneAC ~architecture for a switched linear system with stability and performance analysis. Section~\ref{sec:AC4L2F} introduces the design of the L2F control system, including both the baseline and adaptive controllers, while Section~\ref{sec:flight-test} includes the test results on a UAV.

The notations are defined as follows. 
We denote by $\mathbb{R}^n$  the $n$-dimensional real vector space,  and $\mathbb{R}^{m\times n}$ the set of real $m$ by $n$ matrices. Let $\mbI_n$ denote an $n$ by $n$ identity matrix, and $0$ denote the scalar zero or a zero matrix of compatible dimension. 
We use $\laplace{x(t)}$ to denote the Laplace transform of a signal $x(t)$.  We use $\norm{\cdot}$ to  denote the $2$-norm of a vector or matrix. The maximal and minimal singular values of a matrix $P$ are denoted by $\bar \lambda (P)$ and $\underline \lambda(P)$, respectively. For symmetric matrices $P$ and $Q$, $P>Q$ means $P-Q$ is positive definite. 

\section{Overview of the Learn-to-Fly framework}\label{sec:l2f-overview}
Learn-to-Fly (L2F) is an advanced technology development effort aimed at merging real-time aerodynamic modeling, learning control, real-time guidance, and other disciplines in order to replace the current iterative vehicle development process with aggressive flight testing of self-learning vehicles (see \cref{fig:convVsL2F}).
Figure \ref{fig:L2fblock} presents the L2F framework, and below we provide a brief overview of its primary components.  For more details, the reader is directed to \cite{Heim2018L2F-overview} and references therein.  
\begin{figure}[htbp]
	\centering
	\includegraphics[width=0.7\linewidth]{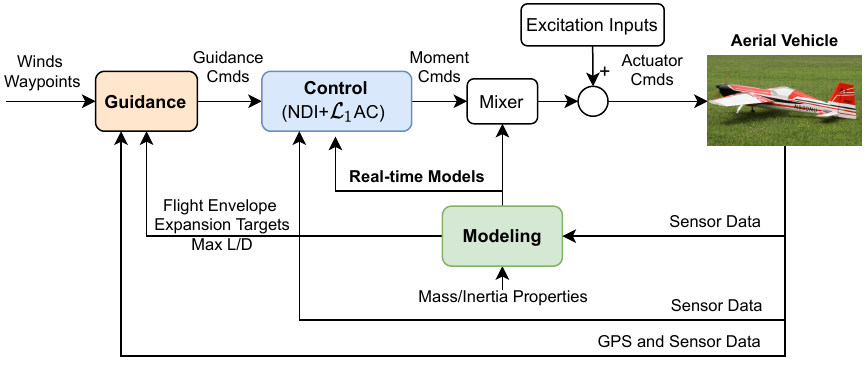}
	\caption{Block diagram of the Learn-to-Fly architecture (reproduction of Fig.~2 in \cite{Heim2018L2F-overview})}
	\label{fig:L2fblock}
\end{figure}

{\bf Modeling}:
One of the main objectives of L2F flight tests is to identify a six-degree-of-freedom mathematical model for each non-dimensional aerodynamic force and moment in real time during the flight.  
The non-dimensional forces and moments must be computed from measured aircraft states, known vehicle geometry, and known mass properties since it is not possible to directly measure these properties during flight \cite{L2FModeling}.
The measured explanatory variables include aircraft states and control surface positions.
Two primary factors in developing an accurate global aerodynamics model are the design of efficient flight maneuvers to sufficiently excite the desired aircraft dynamics and the implementation of a real-time recursive system identification scheme \cite{L2FModeling}.

To provide sufficient excitation to the vehicle dynamics, orthogonal phase-optimized multi-sine inputs were injected into the control surface commands during portions of the flight tests \cite{L2FModeling}.  These inputs are called programmed test inputs (PTIs).  Because the inputs are orthogonal, the effects of each of the inputs can be discerned despite the fact that multiple inputs are excited simultaneously.  Since these inputs are zero-mean and phase-optimized to limit the peak amplitude, their net effect on the aircraft only causes small perturbations from the nominal flight trajectory.

The real-time recursive system identification scheme used was based on multivariate orthogonal functions \cite{L2FModeling}.  
This approach uses a pool of candidate multivariate functions, called regressors, of a user-defined level of complexity.  
In theory, this pool of candidates could be arbitrarily large, but, in practice, the size of the pool is limited by the onboard computing capability \cite{L2FModeling}.  
The explanatory capability unique to each regressor is isolated, quantified, and ranked by orthogonalizing the regressors with a recursive QR decomposition.  
Coefficients for each orthogonalized regressor are calculated using a least-squares estimator that minimizes the sum of squared differences between the computed non-dimensional forces and moments and those predicted by the model.  
The model structure is determined based on measures of model fit quality and a penalty for model complexity.  
These statistical metrics determine which orthogonalized regressors are included in the mathematical model. 
Finally, the model is transformed back from the orthogonal modeling functions to the original physical quantities.  
Uncertainty bounds associated with each model parameter estimate are also computed.  
The multivariate orthogonal function modeling approach is applied to each of the six rigid-body degrees of freedom.
A more detailed discussion of the modeling module is given in \cite{L2FModeling}, where it was first presented.

{\bf Guidance}:
For the flights discussed here, the vehicle was overseen by either an R/C pilot or the real-time guidance algorithm \cite{L2FGuidance}.
An R/C safety pilot would take off and land the aircraft on each flight.  While in control of the vehicle, the R/C pilot would fly the vehicle through various conditions, including stalls, and then engage the autonomous mode.  
In the R/C mode, the control law was not engaged and the pilot's commands are passed directly to the control surfaces.  
In the autonomous mode, the vehicle would fly back and forth between two waypoints following the guidance commands of flight path angle and heading.  

{\bf Control}:
The primary purpose of the control law was to stabilize the vehicle and track the guidance commands.  
The control law had two mechanisms for handling disturbances and uncertainties within the system.  
In the short term, the controller would adapt to these unknowns like many traditional adaptive control architectures.  
However, in the long term, the controller would "learn" and adjust itself based on the real-time modeling results.  
Learning implies modifying the (desired) behavior or performance for the future based on retained knowledge from past experiences.  

In order for the real-time modeling to learn the aerodynamics, a certain amount of perturbation is required, which is provided by the multi-sine inputs described above.  
However, the excitation provided by the multi-sine inputs is at odds with the controller's goal of stabilization.  
In simulation, we found that rejecting too much disturbance hindered the real-time aerodynamic modeling’s ability to identify a model. 
During the L2F process, the conflicting goals of excitation for identification and stability for control must be balanced. 

The controller is composed of a baseline controller and an adaptive augmentation. 
The baseline controller (a nonlinear dynamic inversion (NDI) controller) was designed based on the real-time learned model to  maintain the natural frequency of the vehicle and provide sufficient damping for tracking the guidance commands. 
It uses the onboard aerodynamic model to effectively cancel out any undesirable dynamics. 
The adaptive augmentation assists the inner-loop NDI controller by comparing the measured response to the predicted response of the learned model and adapting the control signal to reduce any discrepancies. 
As mentioned before, periodic update of the learned model naturally leads to a switched reference model for the adaptive control design. The next section addresses this adaptive control design problem.

\section{\lonew adaptive control for uncertain switched linear systems} \label{sec::L1AC}
In this section, we present the design and analysis of an $\mathcal{L}_1$ adaptive controller for a switched linear time-invariant (LTI) system subject to time-varying parametric uncertainties and disturbances. 
It is shown that, with fast adaptation, the closed-loop system with the \lonew controller can behave arbitrarily close to a non-adaptive switched reference system that depends on perfect knowledge of the uncertainties. 
By starting with a stable reference system, the adaptive system is thus ensured to be stable.  
The work in this section differs from previous results \cite{snyder2019switchingL1} in that it adopts the piecewise-constant adaptive law \cite[Section~3.2]{naira2010l1book}, which is more favorable for numerical implementation, as compared to the projection based adaptive law used in \cite{snyder2019switchingL1}. 
Additionally, we explicitly consider the potential (re-)initialization errors associated with the state predictor (e.g., due to measurement inaccuracies) and analyze the performance of the proposed adaptive control systems in the presence of such errors.

%===============================================================================

\subsection{Problem Formulation}
Consider the family of multi-input multi-output LTI systems whose state-space matrices are given by
$
    \{ (A_i, B_i, C_i): i \in \mathcal{I} \}, 
$
where $\mathcal{I}$ denotes the index set. Let $\Sigma = \{ \sigma : [0,\infty) \to \mathcal{I} \}$ denote the set of piecewise-constant switching signals. For a given switching signal $\sigma \in \Sigma$, define the following switched linear system subject to parametric uncertainties and disturbances:
\begin{align}
\begin{split} \label{eq:sysdyn}
   \dot{x}(t) &= A_{\sigma} x(t) + B_{\sigma}\left(\omega_{\sigma} u(t) + \theta^\top_{\sigma}(t)x(t) + d_{\sigma}(t) \right), \\
   y(t) &= C_{\sigma} x(t), \quad x(0) = x_0, 
\end{split}
\end{align}
where $x(t) \in \IR^n$ is the system state, $u(t) \in \IR^m$ is the system input, and $y(t) \in \IR^m$ is the regulated system output,  $A_{\sigma} \in \mathcal{A} \subset \IR^{n \times n}$, $B_{\sigma} \in \mathcal{B} \subset \IR^{n \times m}$, and $C_{\sigma} \in \mathcal{C} \subset \IR^{m \times n}$ are the system matrices, $\omega_{\sigma} \in \Omega \subset \IR^{m \times m}$, $\theta_{\sigma}(t) \in \Theta \subset \IR^{n \times m}$, and $d_{\sigma}(t) \in \Delta \subset \IR^m$ denote the unknown input gain, (time-varying) parametric uncertainties and external disturbances, respectively. 
Given a switching signal $\sigma$, we denote the sequence of finite switching time instant as $t_0,t_1,\ldots,t_i,\ldots$ with $t_0=0$. In the following, we explain the major assumptions. 

\begin{assumption} \label{as:polytopes}
The sets $\mathcal{A}$, $\mathcal{B}$, $\mathcal{C}$, $\Omega$, $\Theta$, and $\Delta$ are  compact, convex polytopes.
We assume that each $B \in \mathcal{B} \subset \IR^{n \times m}$ has full column rank $m (\leq n)$.
Without loss of generality, the sets $\Theta$ and $\Delta$ are assumed to contain $0$.
 $\Omega$ is assumed to be strictly diagonally dominant (which implies non-singularity), and, without loss of generality, is assumed to contain $\II_m$.  
\end{assumption}

The control objective is to design a full-state feedback adaptive control law to ensure that $x(t)$ and $y(t)$ track the ideal signals $x_\textup{id}(t)$ and $y_\textup{id}(t)$ with quantifiable bounds on the transient and steady-state performances. The signals $x_\textup{id}(t)$ and $y_\textup{id}(t)$ are given by the nominal (or desired) dynamics:   
\begin{equation}\label{eq:ideal-dynamics}
\begin{split}
\dot{x}_\textup{id}(t) & =   A_{\sigma} x_\textup{id}(t) + B_{\sigma} u_\idt(t),\ x_\idt(0) = x_0,  \\ u_\idt(t) &= k_{\sigma} r(t),\
y_\textup{id}(t) = C_{\sigma} x_\textup{id}(t),
\end{split}
\end{equation}
where  $r(t)\in \mbR^{m}$ is a given bounded piecewise-continuous reference signal, and $k_{\sigma}\in \mbR^{m\times m}$ is a time-varying feedforward gain for achieving a desired tracking performance. 
In the L2F framework, the ideal system will be determined based on the real-time modeling results and the NDI baseline controller.

\begin{assumption} \label{as:idsysstab}
	There exist symmetric matrices $P_i$ ($i \in \mathcal{I}$) such that 
	\begin{equation} \label{eq:assump-P-lya-stab}
		\begin{split} 
		P_i &\geq \II, \; \forall i \in \mathcal{I}, \\
		A_i^\top P_i + P_i A_i &\leq - \lambda P_i, \; \forall i \in \mathcal{I}, \\
		P_i &\leq \mu P_j, \; \forall i,j \in \mathcal{I},
	\end{split}
	\end{equation}
	for some constants $\lambda > 0$ and $\mu \geq 1$.  Moreover, the switching signal $\sigma$ has a \emph{dwell time}, 
$
		\tau_d \geq \frac{\ln(\mu)}{(1-a^\star)\lambda}, 
$ for arbitrary $a^\star \in (0,1)$. 
	In other words, the switching times $t_1, t_2, \dots$ satisfy the inequality $t_{k+1} - t_{k} \geq \tau_d$ for all $k = 0,1, \dots$.
\end{assumption}

\begin{remark} \label{rem:idsysstab}
	Assumption \ref{as:idsysstab} guarantees that the switched ideal dynamics \eqref{eq:ideal-dynamics} are stable.  
	This can be easily proved following \cite[Chapter 3]{Liberzon} using the switched Lyapunov function defined as $V_{\sigma}(x_\textup{id}(t)) \triangleq x_\textup{id}^\top(t) P_{\sigma} x_\textup{id}(t).$
	Essentially, the Lyapunov function is allowed to switch, at which time it may possibly grow in value by a fixed factor $\mu$.
	By enforcing the dwell time condition based on the growth factor $\mu$ and the decay rate $\lambda$, the value of the Lyapunov will ultimately be decreasing over each switching interval.
	If we remove the dwell time constraint for the switching signal, i.e., allowing for arbitrary switching, a {common} Lyapunov function approach  is usually employed to  guarantee the stability. 
	This can be seen as a special case of dwell-time switching with $\tau_d = 0$, $\mu = 1$. 
	Note that the results derived in this section for dwell-time switching also hold for the more general case of {average-dwell-time switching}. 
	For the details of (average) dwell-time switching, see \cite[Chapter 3]{Liberzon}.
\end{remark}

\subsection{Control architecture}
The proposed control architecture is depicted in Fig.~\ref{fig:l1ac-architecture} and composed of three components, namely, a state predictor, an adaptive law and a low-pass filtered control law. For the subsequent introduction, let $B_{\sigma}^\perp\in \mbR^{n \times (n-m)}$ such that $ B_{\sigma}^\top B_{\sigma}^\perp = 0$ and $\textup{rank}(B_{\sigma}^\vee) = n$, where $B_{\sigma}^\vee \triangleq [B_{\sigma} \, B_{\sigma}^\perp ]$. 
\begin{figure}[h]
\vspace{-2mm}
    \centering
    \includegraphics[width=0.5\textwidth]{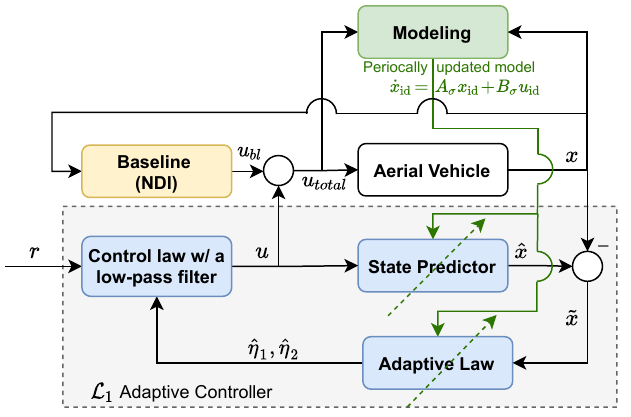}
    \caption{Proposed \loneAC~control architecture for L2F}
    \label{fig:l1ac-architecture}
\end{figure}

\noindent \textbf{State predictor:} The state predictor is given by
\begin{align} \label{eq:state-predictor}
\begin{split}
    &\dot{\hat{x}}(t) = A_{\sigma} \hat{x}(t) + B_{\sigma} (u(t) + \hat{\eta}_1(t) ) + B_{\sigma}^\perp \hat{\eta}_2(t), \\
    & \hat{x}(t_i) = \hat{x}_i,
\end{split}
\end{align}
where $\hat{x}_i$ is the reset value of the states, after switching to the sub-system $i$, and $\hat{\eta}_1(t)$ and $\hat{\eta}_2(t)$ are the estimated uncertainties.

\noindent {\bf Adaptive Law:} The estimation of the (lumped) uncertainties $\hat{\eta}_1(t)$ and $\hat{\eta}_2(t)$ is updated by the following piecewise-constant adaptive law \citep[Section~3.3]{naira2010l1book}:
\begin{align}
\begin{split} \label{eq:adaptive-law}
    \begin{bmatrix} \hat{\eta}_1 (t) \\ \hat{\eta}_2(t) \end{bmatrix} &= \begin{bmatrix} \hat{\eta}_1 (jT_s) \\ \hat{\eta}_2(jT_s) \end{bmatrix}, \quad \forall t \in [jT_s, \, (j+1)T_s), \\
    \begin{bmatrix} \hat{\eta}_1 (jT_s) \\ \hat{\eta}_2(jT_s) \end{bmatrix} &= \begin{bmatrix}\II_m & 0 \\ 0 & \II_{n-m} \end{bmatrix} \left(B_{\sigma}^\vee\right)^{-1} A_{\sigma} \left( e^{-A_{\sigma}T_s} - \II \right)^{-1}  \tilde{x}(jT_s),
\end{split}
\end{align}
where $T_s$ represents the estimation sampling time and $\tilde{x} \triangleq \hat{x} - x$ is the prediction error. 

\noindent {\bf Control Law:} 
The control input is computed through a feedback connection defined, in the Laplace domain, by
\begin{align} \label{eq:l1-control-law}
    u(s) = -\frac{D_0(s)}{s} \mu(s),
\end{align}
where $\mu(s)$ is the Laplace transform of $\mu(t)\trieq u(t) + \hat{\eta}_1(t) - k_{\sigma} r(t)$, and $D_0(s)$ is used to define a series of stable low-pass filters (defined later in \eqref{eq:low-pass-filter}) that sufficiently guarantee the stability of a reference system introduced in Section \ref{sec:sub-stab-perf-analysis}.

\subsection{Stability and performance analysis} \label{sec:sub-stab-perf-analysis}
In this section, we will analyze the stability and performance of the adaptive closed-loop system with the controller defined in \cref{eq:state-predictor,eq:adaptive-law,eq:l1-control-law}, and derive performance bounds for both transient and steady-state phases. For this purpose, we first introduce a non-adaptive reference system that depends on the actual uncertainties (and thus is non-implementable), defined as
\begin{align}
\begin{split} \label{eq:refsys}
   \dot{x}_\reft(t) &= A_{\sigma} x_\reft(t) + B_{\sigma} \big( \omega_{\sigma} u_\reft(t) + \theta^\top_{\sigma}(t) x_\reft(t) + d_{\sigma}(t) \big), \quad x_\reft(0) = x_0, \\
   u_\reft(s) &=   - \frac{D_0(s)}{s}  \mu_\reft(s), 
\end{split}
\end{align}
where $\mu_\reft(s)$ is the Laplace transform of $\mu_\reft(t) \triangleq \omega_{\sigma} u_\reft(t) + \theta^\top_{\sigma} (t) x_\reft(t) + d_{\sigma}(t)- k_{\sigma} r(t)$.
The last equation in \eqref{eq:refsys} is equivalent to\footnote{We use the input-output mapping form instead of a transfer function form in \eqref{eq:uref-equi} since the mapping is time-varying due to the existence of switching.} 
\begin{equation}\label{eq:uref-equi}
    u_\reft = - \omega_{\sigma}^{-1} \mcF_{\sigma} \left( \xi_\reft -  r_k \right), 
\end{equation}
where $\xi_\reft(t) \triangleq \theta^\top_{\sigma} (t) x_\reft(t) + d_{\sigma}(t)$, $r_k(t)\trieq k_{\sigma}r(t)$, and $\mcF_{\sigma} \triangleq \mcF_{i}$ when $\sigma(t) = i$. The (time-invariant) mapping $\mcF_i$ has the transfer function form of
\begin{equation} \label{eq:low-pass-filter}
   \mcC_i(s) \triangleq \omega_i (s\II+D_0(s)\omega_i)^{-1}D_0(s)
\end{equation}
which is a low-pass filter with DC gain equal to an identity matrix, i.e., $\mcF_i(0) = \II_m$. 
From \eqref{eq:refsys} and \eqref{eq:uref-equi}, one can see that for each subsystem $i$, the reference control input tries to cancel the uncertainties within the bandwidth of the filter $\mcF_i(s)$.

Since the reference control law utilizes the true uncertainties, it is not implementable. The reference system is only introduced to help characterize the performance of the adaptive system by providing the {ideal achievable} performance. With the reference system,  the performance of the adaptive system can be analyzed and quantified in four steps: (i) establish the stability of the reference system; (ii) quantify the difference between the state and input signals of the adaptive closed-loop system and those of the reference system; (iii) quantify the difference between the state and input signals of the reference system and those of the ideal system defined in \eqref{eq:ideal-dynamics}; (iv) based on the results from (ii) and (iii), quantify the difference between the state and input signals of the adaptive system and those of the ideal system. Due to space limitations, we focus on the first two steps, while the last two steps are straightforward following existing results \cite{naira2010l1book,zhao2020RALPV}. To simplify the analysis, we also make the following reasonable assumption. 
\begin{assumption} \label{as::Ts}
  The system switches take place at multiples of the sampling period, i.e., each switching time satisfies $t_i = kT_s$ for some $k$.
\end{assumption}
\begin{remark}
   Within the L2F framework, the time instants for model switching are pre-determined, and therefore can be easily selected to satisfy the above assumption.
\end{remark}

\subsubsection{Stability of the reference system}
Letting $(A_f,B_f,C_f,D_f)$ be a minimal realization of $D_0(s)$ with $n_f$ states, the reference system dynamics in \eqref{eq:refsys} can be rewritten in a state-space form as
\begin{align}
\begin{split} \label{eq:refsys-ss}
   \dot{\bar{x}}_\reft(t) &= \bar{A}_{\sigma} \bar{x}_\reft(t) + \bar{B}_{\sigma} d_{\sigma}(t) + \bar{E}_{\sigma} r(t), \quad \bar{x}(0) = \bar{x}_{0}, \\
   u_\reft(t) &= \bar{C} \bar{x}_\reft(t),
\end{split}
\end{align}
where $\bar{x}_\reft(t) \triangleq [x_\reft^\top(t), \; x_{f_1}^\top(t), \; x_{I_1}^\top(t)]^\top$, $\bar{x}_\reft^\top(0) \triangleq [x_0^\top, \; 0, \; 0]^\top$, $x_{f_1}(t)\in\mbR^{n_f}$ and $x_{I_1}(t)\in\mbR^{m}$ are the states of $D_0(s)$ and the integrator in \eqref{eq:refsys}, respectively, and 
\begin{align}\label{eq:A-B-E-C-bar_p-defn}
\begin{split}
 \bar{A}_{\sigma} & \trieq     \left[ \begin{array}{c|c c} A_{\sigma}+B_{\sigma}\theta^\top_{\sigma} & 0 & -B_{\sigma} \omega_{\sigma} \\ \hline  B_f \theta^\top_{\sigma} & A_f & B_f \omega_{\sigma} \\ D_f \theta^\top_{\sigma} & C_f & D_f \omega_{\sigma} \end{array} \right],\ 
  \bar{B}_{\sigma}  \trieq  {\begin{bmatrix} B_{\sigma} \\ B_f \\ D_f \end{bmatrix}},\  \bar{E}_{\sigma} \trieq \begin{bmatrix} 0 \\ B_f k_{\sigma} \\ D_f k_{\sigma} \end{bmatrix},\ 
\bar{C} \trieq \begin{bmatrix} 0 & 0 & -\II \end{bmatrix}.  
\end{split}
\end{align}
\begin{assumption} \label{as:refsysstab}
The parameters of $D_0(s)$ are designed such that there exist symmetric matrices $ \bar{P}_i(\omega)$ ($i \in \mathcal{I}$) such that for all  $(\theta, \omega) \in \Theta \times \Omega$ one has
\begin{equation}\label{eq:switching_stability_condition}
\begin{split}
    \bar{P}_i&\geq \II, \ \forall i\in \mathcal{I}, \\
    \bar{A}_i^\top \bar{P}_i + \bar{P}_i \bar{A}_i   &\leq -\lambda \bar{P}_i, \ \forall i\in \mathcal{I},  \\
     \bar{P}_i &\leq \mu \bar{P}_j,\  
    \forall i,j, \in \mathcal{I},
\end{split}
\end{equation}
for some constants $\lambda>0$ and $\mu\geq 1$.  Moreover, the switching signal satisfies the dwell time constraint 
\begin{equation}\label{eq:dwell-time-constraint}
    \tau_d \geq \frac{\ln(\mu)}{(1-a^\star)\lambda}
\end{equation}
 for arbitrary $a^\star \in (0,1)$.
\end{assumption}

\begin{remark}
Assumption \ref{as:refsysstab} ensures that the control parameters of $D_0(s)$ are designed to stabilize the reference system (following \cite{Liberzon} with the switched Lyapunov function $V_\sigma(\bar{x}_\reft(t)) \triangleq \bar{x}_\reft(t)^\top(t) \bar{P}_\sigma(\omega) \bar{x}_\reft(t)$).  
As the gain of $D_0(s)$, corresponding to the bandwidth of the low-pass filter $\mcF_i(s)$ in \eqref{eq:low-pass-filter}, increases toward infinity, $\mcF_i$ tends to become an identity matrix, $\II$, for all $i \in \mathcal{I}$.
Thus, the effect of the uncertainty in the reference system vanishes.
By substituting \eqref{eq:uref-equi} into the dynamics of \eqref{eq:refsys} with $\mcF_\sigma \equiv \II$, we see that the uncertainties are canceled perfectly, and the reference system becomes the ideal system, which is stable under Assumption \ref{as:idsysstab}.
Therefore, by a perturbation argument, theoretically, $D_0(s)$ can always be designed to stabilize the reference system if Assumption~\ref{as:idsysstab}, which guarantees the stability of the ideal system, holds.
\end{remark}

Note that when the reference system is stable, for a given reference signal $r$ and initial condition $x_0$, there exist constants $\rho_r$ and $\rho_{ur}$ such that 
\begin{align*}
	\norm{x_\reft(t)} \leq \rho_r, \\
	\norm{u_\reft(t)} \leq \rho_{ur}, 
\end{align*}
for all $t \geq 0$ and for all admissible uncertainties. These bounds are independent of the adaptive law and state predictor.

\subsubsection{Transient and steady-state performance}
For the subsequent analysis, define $
    D_\theta  \triangleq \max_{\theta \in \Theta} \norm{\theta}, \ 
    D_d  \triangleq \max_{d \in \Delta} \norm{d}, \ 
    D_\omega  \triangleq \max_{\omega \in \Omega} \abs{\trace(\omega - \II)}.
$
Computing the prediction error dynamics from \eqref{eq:sysdyn} and \eqref{eq:state-predictor}, we have
\begin{align} 
\begin{split} \label{eq:ch4PWCxtildedyn}
    \dot{\tilde{x}}(t) &= A_{\sigma} \tilde{x}(t) + B_{\sigma} \left( (\II-\omega_{\sigma}) u + \hat{\eta}_1(t) - \eta_{\sigma}(t) \right) + B_{\sigma}^\perp \hat{\eta}_2(t), \\
    \tilde{x}(t_i) &= \hat{x}(t_i) - x(t_i),
\end{split}
\end{align}
where  $x(t_i)$ is the state of the actual system when the $i$th switching happens at $t_i$, and 
\begin{equation}\label{eq:eta_p-defn}
      \eta_{\sigma}(t) = \theta_{\sigma}^\top(t)x(t) + d_{\sigma}(t).
\end{equation}
For the following derivations, we first define some bounding functions:
\begin{align*}
    \alpha_1(t) \triangleq \max_{i \in \mathcal{I}} \norm{e^{A_i t}},\ 
    \alpha_2(t) \triangleq \max_{i \in \mathcal{I}} \int_0^t \norm{ e^{A_i(t-\tau)} A_i \left(e^{-A_i T_s} - \II \right)^{-1} } d\tau, \ 
    \alpha_3(t) \triangleq \max_{i \in \mathcal{I}} \int_0^t \norm{ e^{A_i(t-\tau)} B_i} d\tau.
\end{align*}
We further let 
\begin{align*}
    \bar{\alpha}_1(T_s) \triangleq \max_{t \in [0, \, T_s]} \alpha_1(t), \quad 
    \bar{\alpha}_2(T_s) \triangleq \max_{t \in [0, \, T_s]} \alpha_2(t), \quad
    \bar{\alpha}_3(T_s) \triangleq \max_{t \in [0, \, T_s]} \alpha_3(t).
\end{align*}
Notice that $\bar{\alpha}_1(T_s)$ and $\bar{\alpha}_2(T_s)$ are bounded.  Moreover, $\bar{\alpha}_3(T_s)$ satisfies
\begin{align} \label{eq:alpha3lim}
    \lim_{T_s \to 0} \bar{\alpha}_3(T_s) = 0.
\end{align}
Hence, for arbitrary $\delta_0 > 0$, $T_s$ can be chosen such that 
\begin{align} \label{eq:ch4Tscond}
    \left(\bar{\alpha}_1(T_s) + \bar{\alpha}_2(T_s) + 1\right)\bar{\alpha}_3(T_s) \left( D_\omega \rho_u + D_\theta \rho + D_\sigma \right) < \delta_0.
\end{align}

\begin{assumption} \label{assump:xtilde_initialization_constraint}
 We further assume that the re-initialization error at each switching instant satisfies
  \begin{equation}\label{eq:initial_error_constraints}
      \norm{\tilde{x}(t_i)} \leq \bar{\alpha}_3(T_s) \left( D_\omega \rho_u + D_\theta \rho + D_\sigma \right).
  \end{equation}
\end{assumption}
\begin{remark}
   Note that the states of the actual system are measurable; therefore, theoretically, the re-initialization error could be made zero since the states of the state predictor can be set to an arbitrary value. In practice, the states of the actual system may not be measured exactly, e.g., due to measurement error/noise. We thus consider non-zero re-initialization error to account for the practical scenario.
\end{remark}

The following lemma establishes a bound on the prediction error, which can be arbitrarily small under some conditions on the estimation sampling time $T_s$ and re-initialization error.
%---------- PWC xtilde Lemma
\begin{lemma} \label{lem::PWCxtilde}
   Assume $\norm{x(t)} \leq \rho$ and $\norm{u(t)} \leq \rho_u$ for all $0 \leq t \leq \tau$ for some positive constants $\rho $ and $\rho_u$. Given an arbitrary positive constant $\delta_0$, if $T_s$ is chosen to satisfy \eqref{eq:ch4Tscond} and Assumption~\ref{assump:xtilde_initialization_constraint} holds, then $\norm{\tilde{x}(t)} < \delta_0$ for all $0 \leq t \leq \tau$.
\end{lemma}
\begin{proof}
  For arbitrary $j \geq 0$, it follows from the prediction error dynamics in \eqref{eq:ch4PWCxtildedyn} that for the $i$th system on the interval $t \in [0, \, T_s]$ with $t_i + jT_s + t \leq \min{(t_{i+1},\tau)}$, we have
\begin{align} \label{eq:xtilde_t}
\begin{split}
    \tilde{x}(t_i + jT_s + t) &= \underbrace{e^{A_i t} \tilde{x}(t_i + jT_s) + \int_0^t e^{A_i(t-\tau)} B_i^\vee \begin{bmatrix} \hat{\eta}_1(\tau) \\ \hat{\eta}_2(\tau) \end{bmatrix} d \tau}_{\zeta_1(t_i + jT_s+t)} + \underbrace{\int_0^t e^{A_i(t-\tau)} B_i \left( (\II-\omega_i)u_i - \eta_i(\tau) \right) d\tau}_{\zeta_2(t_i + jT_s+t)}.
\end{split}
\end{align}
Substituting the adaptive law in \eqref{eq:adaptive-law} and evaluating at $t = T_s$, we get
\begin{align*}
    \zeta_1(t_i + jT_s + T_s) &= e^{A_i T_s} \tilde{x}(jT_s) + \int_0^{T_s} e^{A_i(T_s-\tau)}B_i^\vee \left(B_i^\vee\right)^{-1}A_i \left( e^{-A_iT_s} - \II \right)^{-1} \tilde{x}(jT_s) d\tau = 0 \\
    \zeta_2(t_i + jT_s + T_s) &= \int_0^{T_s} e^{A_i(T_s-\tau)} B_i \left( (\II-\omega_i) u - \eta_i(\tau) \right) d \tau \leq \bar{\alpha}_3(T_s) \left( D_\omega \rho_u + D_\theta \rho + D_\sigma \right) 
\end{align*}
Thus, for arbitrary $j \geq 0$, it follows that $\tilde{x}(t_i + jT_s) \leq \bar{\alpha}_3(T_s) \left( D_\omega \rho_u + D_\theta \rho + D_\sigma \right)$.  
From \eqref{eq:xtilde_t} and substituting the adaptive laws \eqref{eq:adaptive-law}, we obtain 
\begin{align} \label{eq:ch4PWCxtildebnd}
    \norm{\tilde{x}(t_i + jT_s + t)} &\leq \norm{e^{A_i t} \tilde{x}(t_i + jT_s)} + \norm{\int_0^t e^{A_i(t-\tau)} A_i \left(\II - e^{-A_i T_s} \right)^{-1} \tilde{x}(t_i + jT_s) d \tau} \nonumber+ \norm{\int_0^t - e^{A_i(t-\tau)} B_i d_{\sigma}(\tau) d \tau} \nonumber \\
    & \leq \alpha_1(t) \norm{\tilde{x}(jT_s)} + \alpha_2(t) \norm{\tilde{x}(jT_s)} + \alpha_3(t) \left( D_\omega \rho_u + D_\theta \rho + D_d \right) \nonumber\\
       &\leq \left( \bar{\alpha}_1(T_s) + \bar{\alpha}_2(T_s) + 1\right) \bar{\alpha}_3(T_s)  \left( D_\omega \rho_u + D_\theta \rho + D_d \right) \nonumber  < \delta_0,
\end{align}
where the last inequality is due to \eqref{eq:ch4Tscond}.
This bound holds uniformly for all $i \in \mathcal{I}$, for all $j \geq 0$, and for all $t \in [0,T_s]$. Therefore, 
$\norm{\tilde{x}(t)} < \delta_0$ for all $0 \leq t \leq \tau$. The proof is complete.
\end{proof}
%-------------------

\begin{remark} \label{rem::reinit}
  If the re-initialization is not done such that Assumption \ref{assump:xtilde_initialization_constraint} is met, we see from the proof of Lemma \ref{lem::PWCxtilde} that the effect of this re-initialization error disappears after one time step.  A similar observation is made in \lonew adaptive control of linear parameter-varying systems \cite{zhao2020RALPV} and adaptive estimation for nonlinear systems \cite{zhao2020aR-cbf}.  Moreover, by allowing the state predictor to evolve continuously, i.e., no re-initialization, Assumption \ref{assump:xtilde_initialization_constraint} will automatically be satisfied all the time except perhaps at the initial time $t=0$. 
  Note that in the earlier papers with projection-based adaptation, e.g. \cite{cao2008nonzero}, the initialization error results in additional exponentially decaying terms in the performance bounds. In this paper, the proposed adaptive law provides improved performance in the presence of initialization errors by eliminating their effect after a single time step.
\end{remark}

From the prediction error dynamics in \eqref{eq:ch4PWCxtildedyn}, we have that for all $i \in \mathcal{I}$,
\begin{equation}\label{eq:eta_i_bar}
    \tilde{\eta}_i(t) = B_i^ \dagger \left( \dot{\tilde{x}}(t) - A_i \tilde{x}(t) - B_i^\perp \hat{\eta}_2(t) \right) = B_i^ \dagger \dot{\tilde{x}}(t) - B_i^\dagger A_i \tilde{x}(t),
\end{equation}
where $B_i^ \dagger$ is the pseudo-inverse of $B_i$. 
Letting $y_{f_2}(t)$ be the output after passing $B_i^\dagger \dot{\tilde{x}}(t)$ through the filter $\mcF_i(s)$ (defined in \eqref{eq:low-pass-filter}), we get
\begin{align}
    \begin{bmatrix} \dot{x}_{f_2}(t) \\ \dot{x}_{I_2}(t) \end{bmatrix}= \underbrace{\begin{bmatrix} A_f & B_f \omega_i \\ C_f & D_f \omega_i \end{bmatrix}}_{\trieq \bar{F}_i} \begin{bmatrix} x_{f_2}(t) \\ x_{I_2}(t) \end{bmatrix} + \underbrace{\begin{bmatrix} B_f \\ D_f \end{bmatrix}}_{\trieq\bar{B}_f} B_i^\dagger \tilde{x}(t),\quad  
    y_{f_2}(t) = \underbrace{\begin{bmatrix} C_f & D_f \omega_i \end{bmatrix}}_{\trieq\bar{L}_i} \begin{bmatrix} x_{f_2}(t) \\ x_{I_2}(t) \end{bmatrix} +  D_f  B_i^\dagger \tilde{x}(t) - \gamma_i(t), \label{eq:B-dagger-xbar-dot-dyn}
\end{align}
where $x_{f_2}(t)\in\mbR^{n_f}$ and $x_{I_2}(t)\in\mbR^{m}$ are the states of $D_0(s)$ and the integrator in \eqref{eq:low-pass-filter}, respectively, and
$
    \gamma_i(t) = \bar{C}_f e^{\bar{F}_i (t-t_i)} \bar{B}_f B_i^\dagger \tilde{x}(t_i)
$
 with $\bar{C}_f = \left[0 \quad \II\right]$. 
We note that $\gamma_i(t)$ can be bounded by
\begin{align} \label{eq:ch4gammabnd}
    \norm{\gamma_i(t)} \leq \kappa_\gamma e^{\Lambda_{\bar{F}}(t-t_i)}  \norm{ \tilde{x}(t_i)},
\end{align}
where 
$
    \kappa_\gamma = \max_{i \in \mathcal{I}} \norm{\bar{C}_f \bar{B}_f B_i^\dagger}$ and $ \Lambda_{\bar{F}} = \max_{i \in \mathcal{I}} \lambda_{max}(\bar{F}_i),
$ with $\lambda_{max}(\cdot)$ representing the largest eigenvalue.  

We next characterize the difference between the responses of the adaptive closed-loop system and of the reference system.  Note that $u$ in \eqref{eq:l1-control-law} can be equivalently written as
\begin{equation}\label{eq:control-law-filter-form}
	u = - \omega_{\sigma}^{-1} \mcF_{\sigma} (\eta_{\sigma}+\tilde{\eta}_{\sigma}\ -  k_{\sigma}r), 
\end{equation}
where $\eta_{\sigma}(t)$ is defined in \eqref{eq:eta_p-defn}, $\mcF_{\sigma}$ is introduced in \eqref{eq:uref-equi} and $\tilde{\eta}_{\sigma}(t) \trieq \hat{\eta}_1(t) - \eta_{\sigma}(t) - (\omega_{\sigma} - \II) u(t)$. 
Applying the control law \eqref{eq:control-law-filter-form} to the system dynamics \eqref{eq:sysdyn}, we can rewrite the adaptive closed-loop system similarly to \eqref{eq:refsys-ss} as
\begin{equation}\label{eq:sysdyn-ss}
    \begin{split}
    \dot{\bar{x}}(t) &= \bar A_{\sigma} \bar{x}(t) + \bar B_{\sigma}d_\sigma(t) + \bar E_{\sigma} r(t) + \begin{bmatrix} 0 & B_f^\top & D_f^\top  \end{bmatrix}^\top \tilde{\eta}_{\sigma}(t), \\
    u(t) & = \bar{C} \bar{x}(t),
\end{split}
\end{equation}
where $\bar{x}(t) \triangleq [x^\top(t), \; x_{f_1}^\top(t), \; x_{I_1}^\top(t)]^\top$, $\bar{x}^\top(0) \triangleq [x_0^\top, \; 0, \; 0]^\top$, $x_{f_1}$ and $x_{I_1}$ have the same meaning as in \eqref{eq:refsys-ss}, $\bar A_{\sigma},\ \bar B_{\sigma}$, $\bar C$ and $\bar E_{\sigma}$ are defined in \eqref{eq:A-B-E-C-bar_p-defn}.

The difference between the state and input of the actual system and those of the reference system is denoted as $e(t) \trieq x_\reft(t) - x(t)$ and $e_u(t) \trieq u_\reft(t) - u(t)$. With the definition of 
\begin{equation}\label{eq:bar_e_defn}
	\bar{e}(t) \triangleq [e^\top(t), \; x_{f_1}^\top(t), \; x_{I_1}^\top(t)]^\top, 
\end{equation}
the equations \eqref{eq:refsys-ss} and \eqref{eq:sysdyn-ss} imply that
\begin{equation} \label{eq:ch4edyn}
\begin{split}
      \dot{\bar e}(t) &= \bar A_{\sigma} {\bar e}(t)  - \begin{bmatrix} 0 & B_f^\top & D_f^\top  \end{bmatrix}^\top \tilde{\eta}_{\sigma}(t), \quad  \bar e(0) =  0, \\
  e_u(t) &= \bar C {\bar e}(t)
\end{split}
\end{equation}
Applying \cref{eq:eta_i_bar,eq:B-dagger-xbar-dot-dyn} to the error dynamics in \eqref{eq:ch4edyn} over the time interval $t \in [t_i, t_{i+1})$, we have
\begin{align}
\begin{split} \label{eq:ch4edyn_i_compact}
   \begin{bmatrix} \dot{\bar{e}}(t) \\ \dot{\bar{x}}_{f_2}(t) \end{bmatrix} &= \begin{bmatrix} \bar{A}_i & \bar{H}_i \\ 0 & \bar{F}_i \end{bmatrix} \begin{bmatrix} \bar{e}(t) \\ \bar{x}_{f_2}(t) \end{bmatrix} + \begin{bmatrix} \bar{J}_i \\ \bar{G}_i \end{bmatrix} \tilde{x}(t) + \begin{bmatrix} \bar{\bar{B}}_i \\ 0 \end{bmatrix} \gamma_i(t), \\
   \begin{bmatrix} e_u(t) \end{bmatrix} &= \begin{bmatrix} \bar{C} & \bar{L}_i \end{bmatrix} \begin{bmatrix} \bar{e}(t) \\ \bar{x}_{f_2}(t) \end{bmatrix} + \begin{bmatrix} -D_f B_i^\dagger \end{bmatrix} \tilde{x}(t),
\end{split}
\end{align}
where $\bar{x}_{f_2}(t) = [x_{f_2}^\top(t), \; x_{I_2}^\top(t)]^\top$, $\bar{e}(t) = [e^\top(t), \; \bar{x}_{f_1}^\top(t)]^\top$, and 
\begin{align*}
\begin{gathered}
   \bar{H}_i = \begin{bmatrix} -B_i C_f & -B_i D_f \omega_i \\ 0 & 0 \\ 0 & 0 \end{bmatrix}, 
   \  \bar{J}_i = \begin{bmatrix} -D_f B_i^\dagger \\ -B_fB_i^\dagger A_i \\ -D_f B_i^\dagger A_i \end{bmatrix}, 
   \  \bar{\bar{B}}_i = \begin{bmatrix} B_i \omega_i \\ 0 \\ 0 \end{bmatrix}, \ 
   \bar{G}_i = -\bar{B}_f B_i^\dagger A_i. 
\end{gathered}
\end{align*}

We partition $\bar{P}_i$ in \eqref{eq:switching_stability_condition} along the same partition as $\bar{A}_i$ in \eqref{eq:A-B-E-C-bar_p-defn}, which leads to 
\begin{align*}
	\bar{P}_i &= \left[ \begin{array}{c|c} P_i & R_i \\ \hline R_i^\top & S_i \end{array}\right].
\end{align*}
Further defining $Q_i \triangleq S_i - R_i^\top P_i^{-1} R_i$, we obtain
\begin{equation}
	\begin{split} \label{eq::switching_Q}
		Q_i&\geq \II, \ \forall i\in \mathcal{I}, \\
		\bar{F}_i^\top Q_i + Q_i \bar{F}_i   &\leq -\lambda Q_i, \ \forall i\in \mathcal{I},  \\
		Q_i &\leq \mu Q_j,\  
		\forall i,j, \in \mathcal{I}.
	\end{split}
\end{equation}
Let a scalar $\nu > 0$ satisfy
\begin{align} \label{eq::nudef}
	- \lambda a  \bar{P}_i + \frac{1}{\nu \lambda a}\bar{P}_i \bar{H}_i Q_i^{-1} \bar{H}_i^\top \bar{P}_i < 0 \quad \forall i \in \mathcal{I}, 
\end{align}
for some $a \in (0,a^\star)$, where $a^\star$ was introduced in \eqref{eq:dwell-time-constraint}. Note that such a $\nu$ always exists since $\bar{P}_i > 0$.

Finally, let
\begin{equation}\label{eq:delta12-defn}
	\begin{split}
		\delta_1 \triangleq \sqrt{\left( \mu \left( 1 - \mu^{\frac{a-a^\star}{1-a^\star}} \right)^{-1}  + 1 \right) \frac{g}{(1-a)\lambda} \left( 1 + \kappa_\gamma^2 \right)} \delta_0,\quad   
		\delta_2 \triangleq \max_{i \in \mathcal{I}} \norm{\begin{bmatrix} \bar{C} & \bar{L}_i/\sqrt{\nu} \end{bmatrix}} \delta_1 + \max_{i \in \mathcal{I}} \norm{D_f B_i^\dagger} \delta_0.
	\end{split}
\end{equation}

We are now ready to state the following theorem which quantifies the performance of the adaptive control system, both transient and steady-state, in terms of its states and inputs, as compared to those of the reference system. 
\begin{theorem} \label{thm:ebnd}
Consider the closed-loop adaptive system with the \lonew~controller defined via \cref{eq:state-predictor,eq:adaptive-law,eq:l1-control-law} and the closed-loop reference system \eqref{eq:refsys}. 
			Suppose that there exist $D_0(s)$ and $\bar{P}_i(\omega)$ ($i\in \mathcal{I}$) and some constants $\lambda>0$ and $\mu\geq 1$ such that the inequalities in \eqref{eq:switching_stability_condition} hold  for all  $(\theta, \omega) \in \Theta \times \Omega$; furthermore, suppose that  the dwell time satisfies \eqref{eq:dwell-time-constraint} and the sampling time $T_s$ satisfies \eqref{eq:ch4Tscond}. Then, for all $t \geq 0$, we have
    \begin{align}
        \norm{\tilde{x}(t)} &< \delta_0, \label{eq:ch4PWCxtildebnd2}\\
        \norm{x(t)} & \leq \rho \triangleq \rho_r + \delta_1, \label{eq:ch4PWCxbnd}\\
        \norm{u(t)} & \leq \rho_u \triangleq \rho_{ur} + \delta_2, \label{eq:ch4PWCubnd}\\
        \norm{x_\reft(t) - x(t)} &\leq  \delta_1, \label{eq:ch4PWCebnd}\\
        \norm{u_\reft(t) - u(t)} & \leq \delta_2,\label{eq:ch4PWCeubnd}
    \end{align}
    where $\delta_0$ is introduced in \eqref{eq:ch4Tscond} and $\delta_i$ ($i=1,2$) are defined in \eqref{eq:delta12-defn}.
\end{theorem}
\begin{proof}
Suppose that one of the bounds \eqref{eq:ch4PWCebnd} or \eqref{eq:ch4PWCeubnd} does not hold. Since $e(t)$ and $e_u(t)$ are continuous with $e(0) = 0$ and $e_u(0) = 0$, there exists a $\tau$ such that
\begin{align} \label{eq:ch4contradictme}
    \norm{e(\tau)} &= \delta_1 \text{ or }\\
    \norm{e_u(\tau)} &= \delta_2,
\end{align}
and
\begin{align*}
    \norm{e(t)} < \delta_1, \quad \norm{e_u(t)} < \delta_2, \quad \forall t \in [0, \tau).
\end{align*}
It follows then that for all $t \in [0,\tau)$, one has
\begin{align*}
    \norm{x(t)} & \leq \rho \triangleq \rho_r + \delta_1, \\
    \norm{u(t)} &\leq \rho_u \triangleq \rho_{ur} + \delta_2.
\end{align*}
Applying Lemma \ref{lem::PWCxtilde}, for all $t \in [0,\tau)$, we have
\begin{align} \label{eq:xtildebnd_lem3}
    \norm{\tilde{x}(t)} < \delta_0.
\end{align}

Let $V_i(t) = \bar{e}^\top(t) \bar{P}_i \bar{e}(t) + \nu \bar{x}_{f_2}^\top(t) Q_{i} \bar{x}_{f_2}(t)$ on the time interval $t \in [t_i,t_{i+1})$, where $\bar{e}$ is defined in \eqref{eq:bar_e_defn}, and $\nu$ comes from \eqref{eq::nudef}.
Differentiating $V_i(t)$ along the system trajectories, we have
\begin{align}
	\dot{V}_i(t) &=  \nu \bar{x}_{f_2}^\top(t) \left( \bar{F}_i^\top Q_i + Q_i \bar{F}_i \right) \bar{x}_{f_2}(t) + 2 \nu  \bar{x}_{f_2}^\top(t) Q_i \bar{G}_i \tilde{x}(t) \nonumber \\
	& \quad + \bar{e}^\top(t) \left( \bar{A}_i^\top \bar{P}_i + \bar{P}_i \bar{A}_i \right) \bar{e}(t) +  2 \bar{e}^\top(t) \bar{P}_i \bar{H}_i \bar{x}_{f_2}(t) + 2 \bar{e}^\top(t) \bar{P}_i \bar{J}_i \tilde{x}(t) + 2 \bar{e}^\top(t) \bar{P}_i \bar{\bar{B}}_i \gamma_i(t) \nonumber \\
	& \leq \begin{bmatrix} \bar{e}^\top (t) & \bar{x}_{f_2}^\top(t) & \tilde{x}^\top(t) & \gamma^\top_i(t) \end{bmatrix} \begin{bmatrix} -\lambda \bar{P}_i & \bar{P}_i \bar{H}_i & \bar{P}_i \bar{J}_i & \bar{P}_i \bar{\bar{B}}_i \\ \bar{H}_i^\top \bar{P}_i & -  \nu  \lambda Q_i &  \nu  Q_i \bar{G}_i & 0 \\ \bar{J}_i^\top \bar{P}_i &  \nu \bar{G}_i^\top Q_i & 0 & 0 \\ \bar{\bar{B}}_i^\top \bar{P}_i & 0 & 0 & 0 \end{bmatrix} \begin{bmatrix} \bar{e}(t) \\ \bar{x}_{f_2}(t) \\ \tilde{x}(t) \\ \gamma_i(t) \end{bmatrix} \nonumber \\
	& \leq - (1-a) \lambda V_i(t) + g \norm{\tilde{x}(t)}^2 + g \norm{\gamma_i(t)}^2, \label{eq::ch4Vdot}
\end{align}
where the last line follows from square completion, and the scalar $g$ is given by
\begin{align*}
	g \triangleq \norm{  \begin{bmatrix} \bar{J}_i^\top \bar{P}_i & \nu \bar{G}_i^\top Q_i \\ \bar{\bar{B}}^\top_i \bar{P}_i & 0 \end{bmatrix} \begin{bmatrix} -\lambda a \bar{P}_i & \bar{P}_i \bar{H}_i \\ \bar{H}_i^\top \bar{P}_i & -\nu \lambda a Q_i\end{bmatrix}^{-1} \begin{bmatrix}\bar{P}_i \bar{J}_i & \bar{P}_i \bar{\bar{B}}_i \\ \nu Q_i \bar{G}_i & 0 \end{bmatrix}}.
\end{align*}
Integrating the last inequality in \eqref{eq::ch4Vdot}, we have:
  \begin{align*}
       V_i(t) &\leq V_i(t_i) e^{-(1-a) \lambda (t-t_i)} + \int_{t_i}^t e^{-(1-a)\lambda (t-\tau)} g \left( \norm{\tilde{x}(t)}^2 + \norm{\gamma_i(t)}^2 \right) d\tau.
  \end{align*}
  Substituting the bound for $\norm{\gamma_i(t)}$ from \eqref{eq:ch4gammabnd} and the bound for $\norm{\tilde{x}(t)}$ from \eqref{eq:xtildebnd_lem3}, we have
  \begin{align}
      V_i(t) & < V_i(t_i) e^{-(1-a) \lambda (t-t_i)} + \int_{t_i}^t e^{-(1-a)\lambda (t-\tau)} g \left(1 + \kappa_\gamma^2 e^{2\Lambda_{\bar{F}}(t-t_i)} \right) \delta_0^2 d\tau \nonumber \\
      &< \mu V_{i-1}(t_i) e^{-(1-a)\lambda(t-t_i)} + \frac{g}{(1-a)\lambda} \left(1 + \kappa_\gamma^2 \right) \delta_0^2. \label{eq:ch4PWClyapbnd}
  \end{align}
  If at any switching time $t_i$
  \begin{align*}
      V_{i-1}(t_i) < \left( 1 - \mu^{\frac{a-a^\star}{1-a^\star}} \right)^{-1} \frac{g}{(1-a)\lambda} \left( 1 + \kappa_\gamma^2 \right) \delta_0^2,
  \end{align*}
  then it follows from \eqref{eq:ch4PWClyapbnd} and the dwell time constraint \eqref{eq:switching_stability_condition} that
  \begin{align*}
      V_i(t_{i+1}) &< \left( 1 + \mu \left( 1 - \mu^{\frac{a-a^\star}{1-a^\star}} \right)^{-1} \mu^{-\frac{1-a}{1-a^\star}}  \right) \frac{g}{(1-a)\lambda} \left( 1 + \kappa_\gamma^2 \right) \delta_0^2 \\
      &< \left( 1 - \mu^{\frac{a-a^\star}{1-a^\star}} \right)^{-1} \frac{g}{(1-a)\lambda} \left( 1 + \kappa_\gamma^2 \right) \delta_0^2~,
  \end{align*}
  and it holds for all $t_j$ with $j \geq i$.  Since $V_i(0) = 0$, it holds for all switching times.  Therefore, from \eqref{eq:ch4PWClyapbnd} we have
  \begin{align*}
      \norm{ \begin{bmatrix} \bar{e}(t) \\ \sqrt{\nu} \bar{x}_{f_2}(t) \end{bmatrix} }^2 \leq V_i(t) & < \left( \mu \left( 1 - \mu^{\frac{a-a^\star}{1-a^\star}} \right)^{-1}  e^{-(1-a)\lambda(t-t_i)} + 1 \right) \frac{g}{(1-a)\lambda} \left( 1 + \kappa_\gamma^2 \right) \delta_0^2 \\
      & < \left( \mu \left( 1 - \mu^{\frac{a-a^\star}{1-a^\star}} \right)^{-1}  + 1 \right) \frac{g}{(1-a)\lambda} \left( 1 + \kappa_\gamma^2 \right) \delta_0^2 = \delta_1^2,
  \end{align*}
  and from the definition of $e_u$ in \eqref{eq:ch4edyn}, we have
  \begin{align*}
      \norm{e_u(t)} < \max_{i \in \mathcal{I}} \norm{\begin{bmatrix} \bar{C} & \bar{L}_i/\sqrt{\nu} \end{bmatrix}} \delta_1 + \max_{i \in \mathcal{I}} \norm{D_f B_i^\dagger} \delta_0 = \delta_2.
  \end{align*}
  This contradicts \eqref{eq:ch4contradictme}, and thus \eqref{eq:ch4PWCebnd} and \eqref{eq:ch4PWCeubnd} hold.  Applying the triangle inequality then gives \eqref{eq:ch4PWCxbnd} and \eqref{eq:ch4PWCubnd}.  Finally, \eqref{eq:ch4PWCxtildebnd2} follows from Lemma \ref{lem::PWCxtilde}.  The proof is complete.
\end{proof}

To summarize, if the reference system is stable and the stated assumptions are met, Theorem \ref{thm:ebnd} guarantees that the adaptive system follows the reference system within some a priori bounds.  
If the re-initialization errors do not satisfy Assumption \ref{as::Ts}, then, as stated in Remark \ref{rem::reinit}, the effect of the re-initialization error only lasts for one time step.  
The size of the bounds in Theorem \ref{thm:ebnd} are dependent on the choice of the time step $T_s$.  
By decreasing $T_s$, under the adaptive control law in \eqref{eq:state-predictor}, \eqref{eq:adaptive-law}, \eqref{eq:l1-control-law}, the system can be made arbitrarily close to the stable reference system.

\section{Controller design for Learn-to-Fly}\label{sec:AC4L2F}
In this section, we introduce the details of the controller design for the L2F system, including how the switched \lonew controller presented in Section \ref{sec::L1AC} is incorporated.  Before proceeding, we present a table listing the nomenclature used in this section.  Any reuse of symbols in this section should not be confused with those used in the derivation and analysis of the \lonew control architecture in Section \ref{sec::L1AC}.

\begin{table}[htb]
	\centering
	\caption{Aircraft Nomenclature}
	\begin{tabular}{c l  c l} \hline\hline
		Symbol & Meaning & Symbol & Meaning  \\ \midrule
		$V$  & Total Airspeed &  $L$ & Lift Force    \\
		$\alpha$  & Angle of Attack  & $Y$ & Side Force  \\ 
		$\beta$  & Angle of Sideslip  & $\phi$  & Roll Angle    \\
		$p$  & Roll Rate &  $\theta$  & Pitch Angle  \\
		$q$  & Pitch Rate   & $\chi$  & Ground Track Angle   \\ 
		$r$  & Yaw Rate  &  $\gamma$ & Flight Path Angle   \\
		$\omega$ & Angular Rate Vector $[p\, q\, r]^\top$ & $\mu$ & Bank Angle  \\
		$M$ & External Aerodynamic Moment Vector   &  $m$ & Vehicle Mass \\ 
		$M_\delta$ & Moment Produced by Control Surfaces  &  $I$ & Inertia Matrix \\
		$g$ &  Acceleration Due to Gravity  &  & \\ \hline\hline
	\end{tabular}
\end{table}

\subsection{Mathematical model}
To begin, we define the mathematical model used to develop the control law.  We start with the standard equations of motion for a rigid-body aircraft \cite{Snell} with no thrust: 
\begin{equation} \label{eq:aircraft_dynamics}
\begin{split}
	\dot{\chi} &= \frac{1}{mV\cos(\gamma)} \left(  L \sin(\mu) + Y \cos(\mu)\cos(\beta) \right), \\
	\dot{\gamma} &= \frac{1}{mV} \left( L \cos(\mu) - Y \sin(\mu) \cos(\beta) - mg \cos(\gamma) \right) , \\
	\dot{\phi} &= p + \tan(\theta) \left( q \sin(\phi) + r \cos(\phi) \right), \\
	\dot{\beta} &= -\cos(\alpha)r+\sin(\alpha)p + \frac{1}{mV} \left( Y \cos(\beta) + mg \cos(\gamma)\sin(\mu) \right), \\
	\dot{\alpha} &= q-\tan(\beta) \left( \cos(\alpha)p + \sin(\alpha) r \right)  + \frac{1}{mV\cos(\beta)} \left( -L + mg\cos(\gamma)\cos(\mu) \right), \\
	\dot{\theta} &= q \cos(\phi) - r \sin(\phi), \\
	\dot{\omega} &= I^{-1} \left( M + M_\delta \right) - I^{-1} \left( \omega \times I\omega \right).
\end{split}
\end{equation}

If we assume small angles ($\alpha$, $\beta$, $\gamma$) and small rates ($p$, $r$) with negligible side force ($Y=0$) and steady 1-g flight ($L = mg/\cos(\phi)$),
the dynamics in \eqref{eq:aircraft_dynamics} can be simplified as 
\begin{align}
	\begin{split}
		\dot{\chi} &= \frac{g}{V}\tan(\phi), \\
		\dot{\gamma} &= 0, \\
		\dot{\phi} &= p + \tan(\theta) \left( q \sin(\phi) + r \cos(\phi) \right), \\
		\dot{\beta} &= -r + \frac{g}{V}\sin(\phi),  \\
		\dot{\alpha} &= q  - \frac{g}{V} \sin(\phi) \tan(\phi), \\
		\dot{\theta} &= q \cos(\phi) - r \sin(\phi), \\
		\dot{\omega} &= I^{-1} \left( M + M_\delta \right) - I^{-1} \left( \omega \times I\omega \right).
	\end{split}
	\label{eq::SimplifiedDynamics}
\end{align}
This simplified model was used to develop the control laws.  
Originally, the baseline control law was designed for an unpowered glider-type aircraft. 
In addition, during these tests, the real-time modeling software was not able to separate the thrust effects from the aerodynamic effects.  
For these reasons, despite the test vehicle having thrust, the control law does not consider its effects.  
The velocity was loosely controlled by a simple proportional-integral law feeding the velocity error back to the throttle.  
The feedback loop was unknown to the controller and can be thought of as unmodeled dynamics.
As an example, if we look at the angle of attack dynamics in the presence of thrust $T$, we can see how the interaction might skew the aerodynamic model:
\begin{align*}
	\dot{\alpha} = & \, q - \tan(\beta) \left( \cos(\alpha)p + \sin(\alpha) r \right) + \frac{1}{mV\cos(\beta)} \left( -L + mg \cos(\gamma) \cos(\mu) - T_x \sin(\alpha) + T_z \cos(\alpha) \right).
\end{align*}
Assuming the thrust is axially aligned (i.e., $T_z = 0$), the modeled lift force will contain both the actual lift force $L$ and a thrust component $T_x \sin(\alpha)$. With the thrust loop closed based on the velocity, the modeled lift force now contains an artificial dependence on speed, angle of attack, and the dynamics of motor thrust.  

Since this initial flight testing, the capabilities of the real-time aerodynamic modeling have expanded to capture thrust effects by including modeling terms based on the advance ratio (similar to the use of corrected engine rotational speed in \cite{Brandon2016}).  
Unfortunately, further flight tests examining the corresponding effect on the performance of the baseline and adaptive controllers have not been performed.

%----------------------------------------------------------------------------------------------------------------
\subsection{Baseline controller design}
As mentioned in Section \ref{sec:l2f-overview}, the baseline controller is a nonlinear dynamic inversion (NDI) controller.  To apply NDI to the dynamics in \eqref{eq::SimplifiedDynamics}, we sequentially generate commands based on the time-scale separation of the variables.  
The first loop converts the guidance commands of ground track angle $\chi$ and flight path angle $\gamma$ into attitude commands for the roll angle $\phi$ and pitch angle $\theta$, respectively.  
The desired derivatives are chosen to be proportional to the error between the variables and its command in order to produce a linear tracking response.  
Inverting the $\chi$ dynamics from \eqref{eq::SimplifiedDynamics} in this manner produces a roll command given by
\begin{align*}
	\phi_{cmd} &= \sat{-45^\circ}{+45^\circ}{\arctan\left(\frac{V}{g} K_\chi (\chi_{cmd} - \chi) \right)},
\end{align*}
where we have additionally limited the roll command between $\pm 45^\circ$.
Since the flight path angle dynamics in \eqref{eq::SimplifiedDynamics} are constant, we use a kinematic relationship between flight path angle $\gamma$ and pitch angle $\theta$ from the $z$-component of the velocity vector to directly solve for a commanded pitch angle:
\begin{align*}
	\sin(\gamma_{cmd}) = a_1 \cos(\theta_{cmd}) + a_2 \sin(\theta_{cmd}),
\end{align*}
where $a_1 = -\cos(\alpha) \cos(\beta)$ and $a_2 = \sin(\phi) \sin(\beta) + \cos(\phi) \sin(\alpha) \cos(\beta)$ \cite{StevenLewis}.

The second loop converts the angle commands into angular rate commands, with $\phi$, $\theta$, and $\beta$ being mapped to commands for $p$, $q$, and $r$, respectively.  Once again seeking a linear tracking response, from \eqref{eq::SimplifiedDynamics}, the commands are
\begin{align*}
	p_{cmd} &= K_\phi (\phi_{cmd} - \phi) - \tan(\theta) \left( q \sin(\phi) + r \cos(\phi) \right),\\
	q_{cmd} &= \frac{1}{\cos(\phi)}\left(  K_\theta (\theta_{cmd} - \theta) + r \sin(\phi) \right), \\
	r_{cmd} &= -K_\beta (\beta_{cmd}-\beta) - \frac{g}{V}\sin(\phi).  
\end{align*}

The final loop converts these angular rate commands into moment commands:
\begin{align*}
	M_{\delta, cmd} &= I K_\omega (\omega_{cmd}-\omega)  -  (\hat{M} - \omega \times I\omega),
\end{align*}
where $K_\omega = \text{diag}([K_p, K_q, K_r])$, $\omega_{cmd} = [p_{cmd}, q_{cmd}, r_{cmd}]^\top$, and an estimate of the current moment on the aircraft $\hat{M}$, as determined by the real-time aerodynamic modeling discussed in Section \ref{sec:l2f-overview} and detailed in \cite{L2FModeling}, is used in place of a measurement.  

\subsection{Model and inflight control law update from online learning}
The control parameters $K_\theta$, $K_q$, $K_\beta$, $K_r$, $K_\phi$, and $K_p$ were selected, based on the updated model, to maintain the natural frequency of the vehicle while supplying sufficient damping.  Assuming the dynamic inversion is capable of removing the nonlinear effects, which would leave only the desired linear dynamics, the longitudinal dynamics would look like
\begin{align}
	\begin{bmatrix} \dot{\theta} \\ \dot{q} \end{bmatrix} &= \underbrace{\begin{bmatrix} 0 & 1 \\ -K_q K_\theta & -K_q \end{bmatrix}}_{A_\sigma} \begin{bmatrix} \theta \\ q \end{bmatrix} + \underbrace{\begin{bmatrix} 0 \\ K_q K_\theta \end{bmatrix}}_{B_\sigma} \theta_{cmd}, \label{eq:experiment_des_dynamics}
\end{align}
which is a canonical description of a second-order system.  The control parameters can then be written in terms of the natural frequency $\omega_n$ and damping $\zeta$ of the second-order system
\begin{align*}
	K_q &= 2 \zeta \omega_n \\
	K_\theta &= \omega_n / 2\zeta.
\end{align*}
A similar description can be given for the roll and yaw dynamics.  

The desired natural frequency was determined based on the linearized real-time modeling results, as described in Table \ref{tab::desiredFreq}, and the desired damping ratio was selected to be $0.8$.  
The non-linear, global non-dimensional model was updated at a rate of 5 Hz by the real-time modeling.  The linearization of the non-dimensional model was then computed based on the desired operating point.  From Table \ref{tab::desiredFreq}, we see that the desired natural frequencies (desired dynamics) were parameterized by the dynamic pressure $\bar{q}$ during dimensionalization.  Given the 5 Hz update rate, we reasonably assumed that the dynamic pressure remains constant over that interval; thus, we consider the system to be comprised of switched linear systems as discussed in Section \ref{sec::L1AC}.  For comparison, the flight control law implementation ran at a rate of 50 Hz.  

To verify the stability of the switched system \cref{eq:experiment_des_dynamics} characterizing the desired dynamics, we need to check whether Assumption~\ref{as:idsysstab} holds. For verification,
we can constrain $\omega_n$ in $\Omega=[0.5,20]$ rad/s, which is sufficiently large to cover all possible scenarios. Suppose we divide $\Omega$ to four subsets, namely, $\Omega^1=[0.5,0.7]$, $\Omega^2=[0.7,1.6]$, $\Omega^3=[1.6,5]$ and $\Omega^4=[5,20]$, and assume that $\omega_n$ will not jump between non-adjacent subsets, e.g., $\Omega^1$ and $\Omega^3$. 
Then, we could find four constant matrices, namely,
\begin{equation}
    P^1= \begin{bmatrix}
           68.89 & 100.3\\ 100.3 &  235.7
    \end{bmatrix},\ 
P^2= \begin{bmatrix}
           99.64 & 50.35\\ 50.35 & 65.42
    \end{bmatrix},\ 
P^3= \begin{bmatrix}
       140.63 & 26.91\\ 26.91 & 17.25
\end{bmatrix},\ 
P^4= \begin{bmatrix}
      452.5 & 24.15\\ 24.15&4.67
\end{bmatrix},
\end{equation}
with $P^i$ corresponding to $\Omega^i$, 
together with $\mu=4.5$ and $\lambda=0.7$,
such that the condition in Assumption~\ref{as:idsysstab} holds, thereby guaranteeing the stability of \cref{eq:experiment_des_dynamics}, as long as the dwell time of the switching signal is larger than $\tau_a = \ln(4.5)/0.7=2.15$ second.
However, if we consider a dwell time of 0.2 second for the switching signal (determined by the model update frequency of 5 Hz), 
we were not able to find matrices $P_i$ that satisfy the condition \cref{eq:assump-P-lya-stab} in Assumption~\ref{as:idsysstab} , which indicates that the stability of \cref{eq:experiment_des_dynamics} under a dwell time of 0.2 second cannot be theoretically verified. The lack of
theoretical verification for the stability of the desired dynamics used in the flight tests could be attributed to a few factors. 
First, the range for the desired natural frequency $\omega_n$, selected as $[0.5,20]$ to cover all possible scenarios,  could be overly large. 
Second, the condition \cref{eq:assump-P-lya-stab} in Assumption~\ref{as:idsysstab} is merely a sufficient
condition for stability guarantee (based on the Lyapunov theory) and could be quite conservative in practice. 
Third, the condition \cref{eq:assump-P-lya-stab} in Assumption~\ref{as:idsysstab} guarantees the stability under {\it continual} switching, while the desired dynamics in the flight tests will converge in a few seconds after the learning starts (as shown in \cref{fig:F11Pilot,fig:F10Angles,fig:F11Angles}), and no switching will happen afterward.
As noted in Remark~\ref{rem:idsysstab}, the conditions of Assumption~\ref{as:idsysstab} could be replaced with average dwell time conditions, which are easier to satisfy.

\begin{table}
	\renewcommand{\arraystretch}{2.0}
	\caption{Desired frequency calculations}
	\centering
	\begin{tabular}{  l  c  }
		\hline\hline
					
		Degree of freedom & Equation \\
		\hline
		\hspace{0.25in} Roll \hspace{0.25in}  &  \hspace{0.25in} $\omega_n = \sqrt{ \left| \frac{\bar{q}Sb}{2 I_{xx}} C_{l_{\delta_a}} \right| }$ \hspace{0.25in} \\  
		\hspace{0.25in} Pitch \hspace{0.25in} &  \hspace{0.25in} $\omega_n = \sqrt{ \left| \frac{\bar{q}S\bar{c}}{ I_{yy}} C_{m_{\alpha}} \right| }$ \hspace{0.25in}   \\  
		\hspace{0.25in} Yaw \hspace{0.25in} &  \hspace{0.25in} $\omega_n = \sqrt{ \left| \frac{\bar{q}Sb}{ I_{zz}} C_{n_{\beta}} \right| } $ \hspace{0.25in}    \\  \hline \hline
	\end{tabular}
	\label{tab::desiredFreq}
\end{table}

%-----------------------------------------------------------------------------------
\subsection{\lonew adaptive augmentation}

The adaptive portion of the controller is used to assist the control loop regulating the angular rates.  
The desired dynamics are defined as the dynamic inversion component successfully canceling the nonlinearities and enforcing the ideal linear response: 
	\begin{align*}
		\dot{\omega}(t) = -K_\omega \omega(t) + K_\omega \omega_{cmd}(t).
	\end{align*}							   
This can be represented in the form of \eqref{eq:ideal-dynamics}:
\begin{align*}
	\dot{x}_\textup{id}(t) &= A_{\sigma} x_\idt(t) + B_{\sigma} u_\idt(t), \quad x_\idt(0) = x_0, \\
	u_\idt(t) &= k_{\sigma} r(t), \quad y_\idt(t) = C_{\sigma} x_\idt(t),
\end{align*}
where we now transition back into the adaptive control nomenclature with the state $x$ representing the angular velocity $\omega$, and the reference $r$ representing the commanded angular velocity $\omega_{cmd}$.  The system matrices $A_{\sigma}$ and $B_{\sigma}$ are given by $-K_\omega$ and $\II$, respectively, and $k_{\sigma}$ and $C_{\sigma}$ are equal to $K_\omega$ and $\II$, respectively.  The values for $K_\omega$ are updated (switched) based on the real-time modeling results, as discussed in the previous subsection.

Due to the simplifying assumptions, model uncertainties, external disturbances, etc., the dynamic inversion controller will not be able to exactly achieve its goal.  
We thus arrive at the uncertain dynamics considered in \eqref{eq:sysdyn} and apply the adaptive controller discussed in Section \ref{sec::L1AC}.

\section{Flight test results}\label{sec:flight-test}
The flight tests described here were performed by NASA Langley Research Center on the E1 flight test vehicle, shown in Figure \ref{fig:E1}.  The vehicle is able to be flown in three modes: safety mode, pilot command pass-through mode, and full autonomy mode. In the safety mode, the vehicle is remotely piloted in a stick-to-surface fashion.  This allows the R/C pilot to takeoff and land the vehicle, as well as to intervene at any time in order to preserve the safety of the vehicle, test crew, or airspace.  The pilot command pass-through mode allows signals to be injected on top of the R/C pilot commands.  In the nominal flight configurations, this allows the introduction of the PTIs during piloted operation.  For the destabilized configurations, it is also used to introduce a destabilizing feedback loop to redundant control surfaces.  For more details on the destabilization and the modes in general, see \cite{L2FOps}.  Finally, the full autonomy mode engages the onboard guidance, modeling, and control algorithms with the injected PTIs and any destabilizing feedback inputs.  The pilot has no control of the vehicle in this mode.  

\begin{figure}[htbp]
	\centering
	\includegraphics{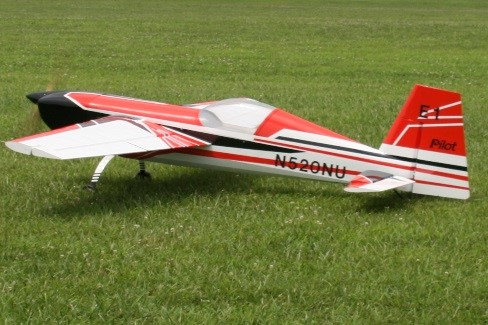}
	\caption{E1 Flight Test Vehicle}
	\label{fig:E1}
\end{figure}

A brief summary of the what was varied between flights is included in Table \ref{tab:FlightTestSummary}, along with any corresponding figure numbers.  A more detailed description of the flights and discussion of the flight data is provided in the next two subsections.  

\begin{table} [htb]
	\centering
	\caption{Flight Test Summary} \label{tab:FlightTestSummary}
	\begin{tabular}{l l l l}
		Configuration & Controller & Learning & Figures \\ \hline
		Nominal Dynamics & NDI Only & Sequential & Fig. \ref{fig:F2Angles} \\
		Nominal Dynamics & NDI+\lonew & Sequential & Fig. \ref{fig:F3Angles}, \ref{fig:F3Stall} \\
		Pitch Destabilized (-10\% static margin) & NDI+\lonew & Simultaneous & No Fig. \\
		Pitch Destabilized (-16.4\% static margin) & NDI+\lonew & Simultaneous & Fig. \ref{fig:F11Angles}, \ref{fig:F11uAd}, \ref{fig:F11Pilot} \\
		Roll Destabilized & NDI+\lonew & Simultaneous & Fig. \ref{fig:F10Angles}
	\end{tabular}
\end{table}

\subsection{Control performance in nominal flight conditions}
In the first flight test experiment we examined the performance of the NDI controller with and without the \loneAC ~law.  This was a qualitative comparison of the overall L2F system with and without adaptive element of the controller in nominal conditions when there are no or small modeling errors.  Due to the different real-time modeling results that were available to the two controllers during their respective flights, a true apples-to-apples comparison of the control methods cannot be made with the collected data.  Instead, a carefully designed future experiment would be required to judge the performance of each control methodology.  

\emph{Without \lonew compensation}:  The first flight of the experiment began with an R/C takeoff.  The R/C pilot engaged the learning mode with the PTIs once reaching the appropriate altitude.  Data were collected over three laps of a racetrack pattern, two near the target speed and one at a lower speed.  The PTIs were disengaged and the pilot performed pitch, yaw, and roll doublets.  From here, the autopilot with only the NDI controller (i.e., no adaptive component) was given the control authority.  There were no significant transients during this handoff.  The vehicle began following the waypoints in the racetrack pattern.  After demonstrating sufficient control of the vehicle, the autopilot was disengaged and reengaged multiple times at various headings, bank angles, and speeds.  The results were similar for each case.  Finally, the vehicle was returned to R/C mode for landing.  

For one of the passes, Figure \ref{fig:F2Angles} shows the system states in blue and the commanded values in red.  Good tracking performance can be seen in both the roll and pitch channels, but the yaw channel shows a small bias, especially in the sideslip angle $\beta$.  
Note that here, and throughout, we report the vehicle speed as the calibrated airspeed $V_{CAS}$.

%---------------------------------- F2 Figures
\begin{figure}[htbp]
	\centering
	\includegraphics[trim={0.in 0in 0.4in 0.2in},clip,width = 0.32\linewidth]{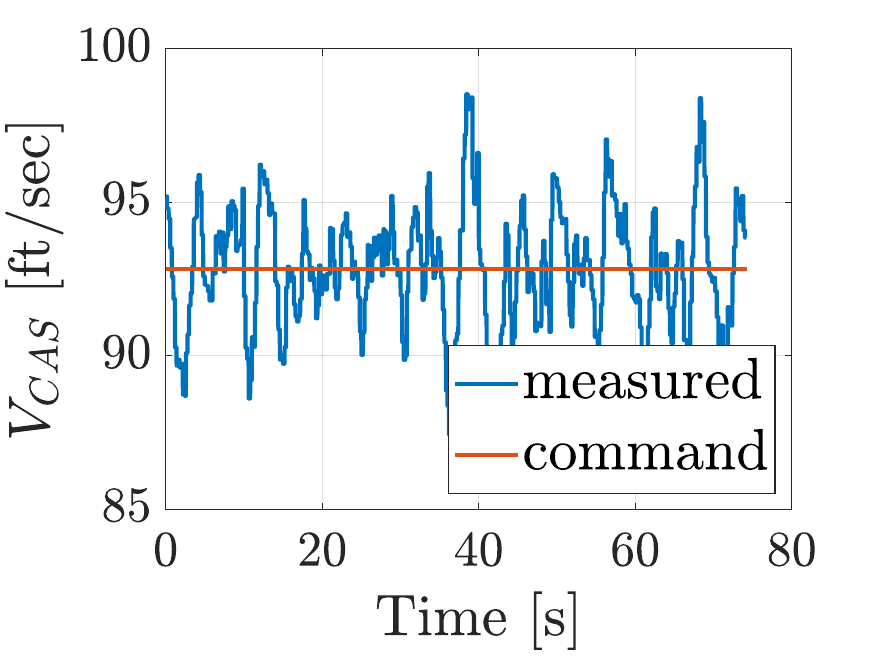}
	\includegraphics[trim={0.in 0in 0.4in 0.2in},clip,width = 0.32\linewidth]{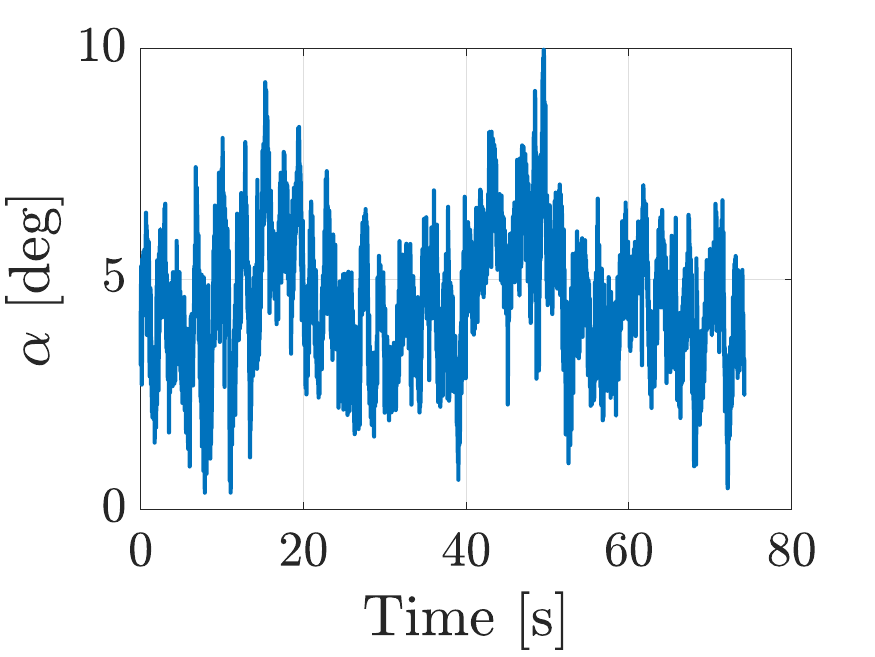}
	\includegraphics[trim={0.in 0in 0.4in 0.2in},clip,width = 0.32\linewidth]{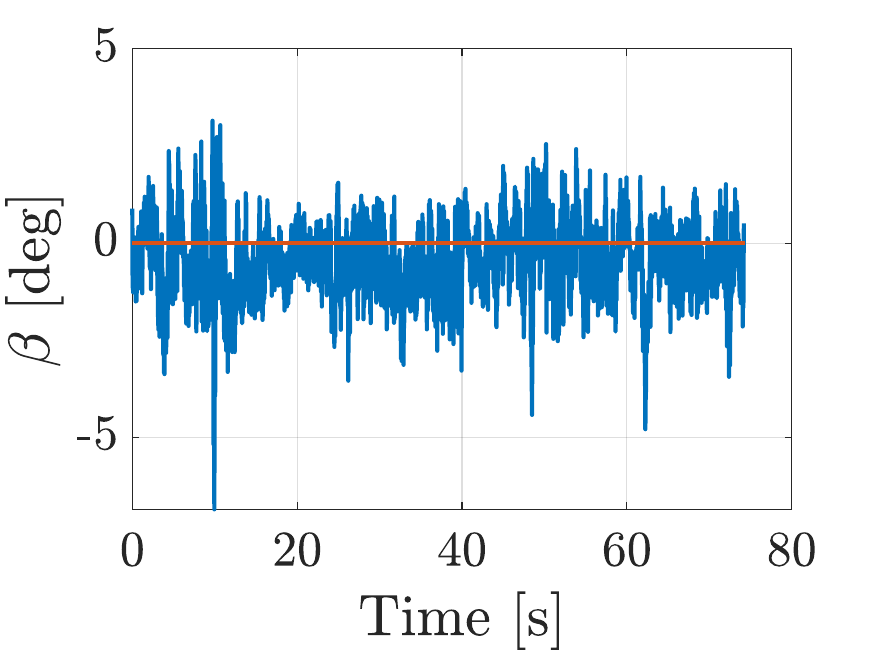}
	\includegraphics[trim={0.in 0in 0.4in 0.2in},clip,width = 0.32\linewidth]{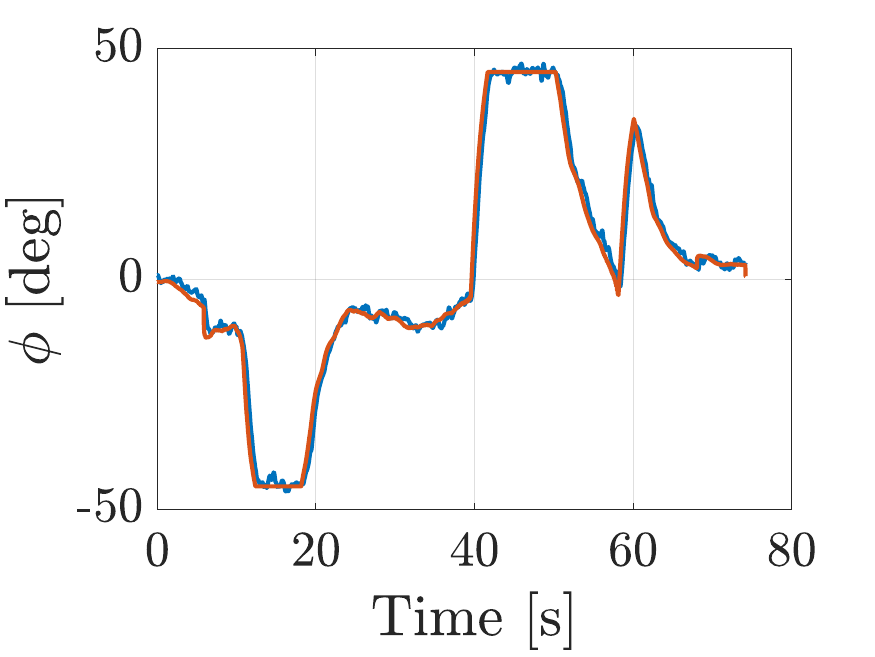}
	\includegraphics[trim={0.in 0in 0.4in 0.2in},clip,width = 0.32\linewidth]{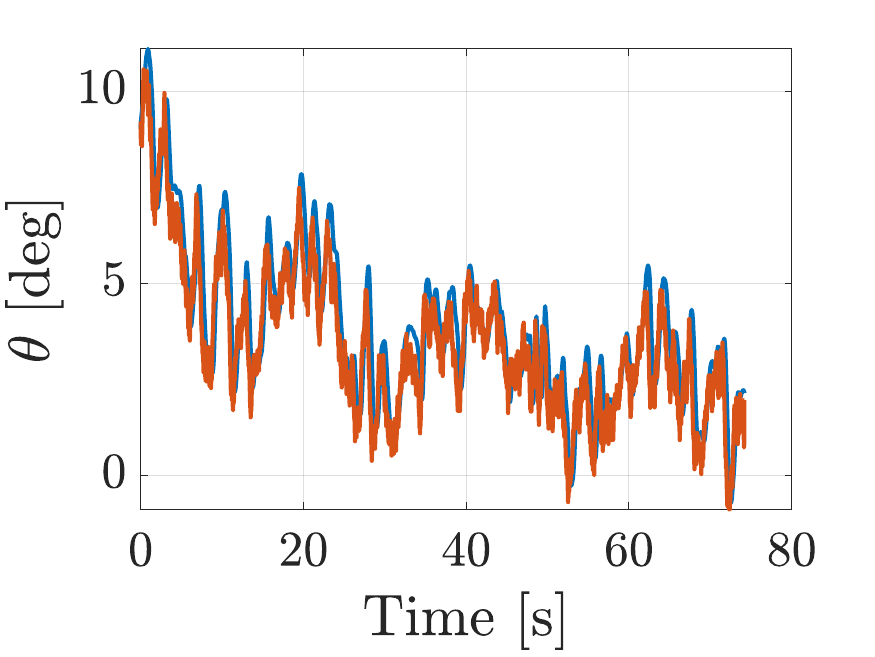}
	\includegraphics[trim={0.in 0in 0.4in 0.2in},clip,width = 0.32\linewidth]{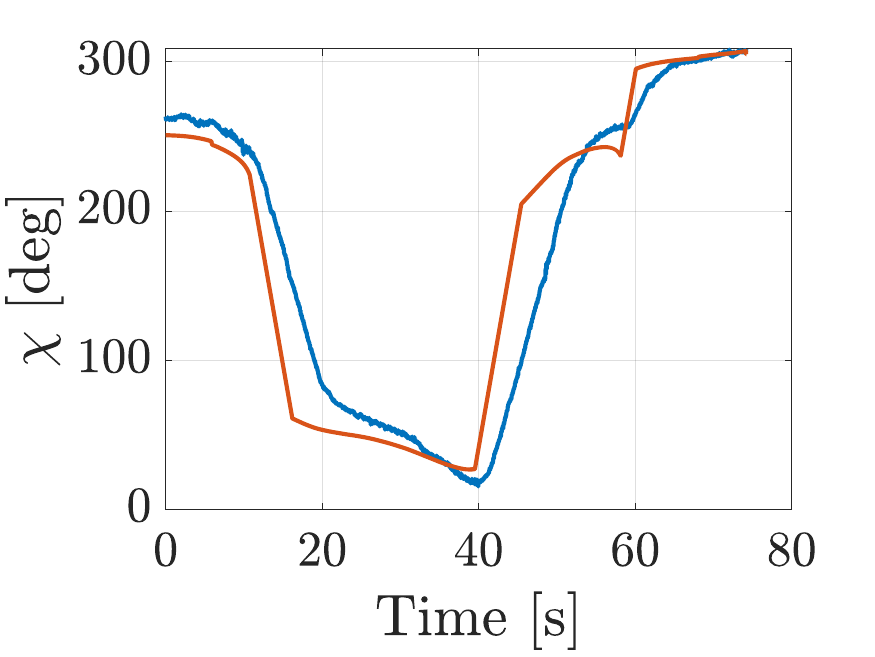}
	\includegraphics[trim={0.in 0in 0.4in 0.2in},clip,width = 0.32\linewidth]{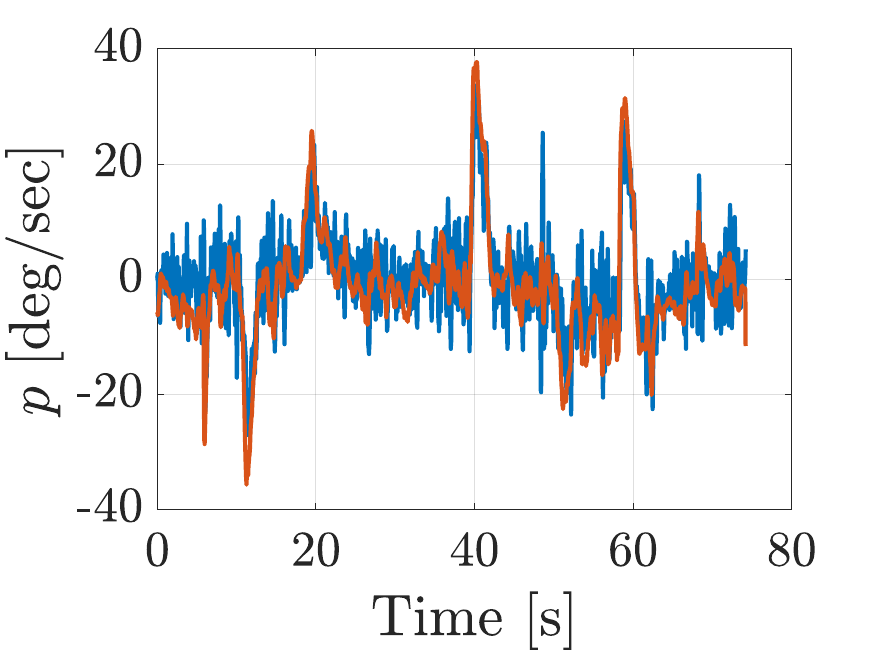}
	\includegraphics[trim={0.in 0in 0.4in 0.2in},clip,width = 0.32\linewidth]{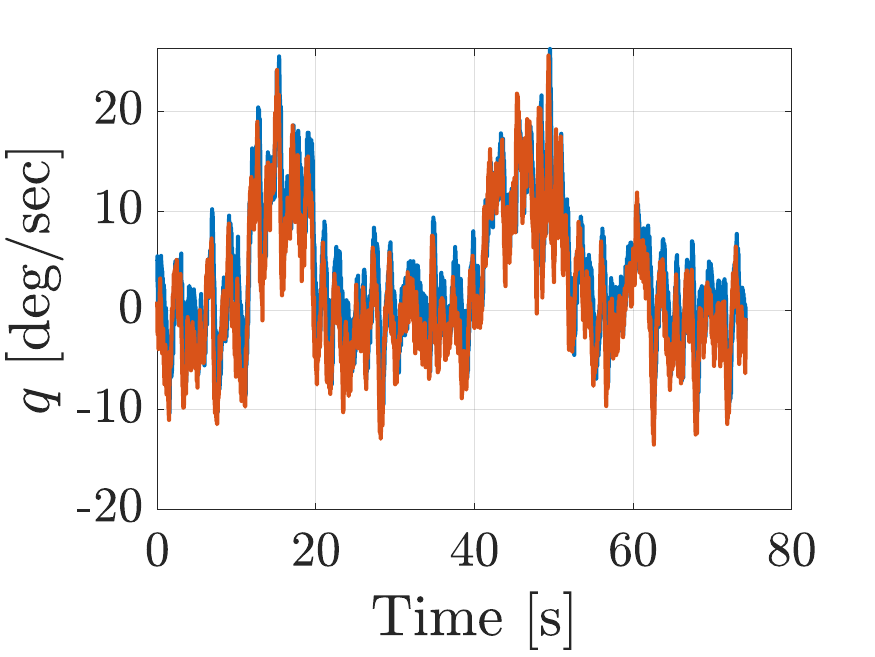}
	\includegraphics[trim={0.in 0in 0.4in 0.2in},clip,width = 0.32\linewidth]{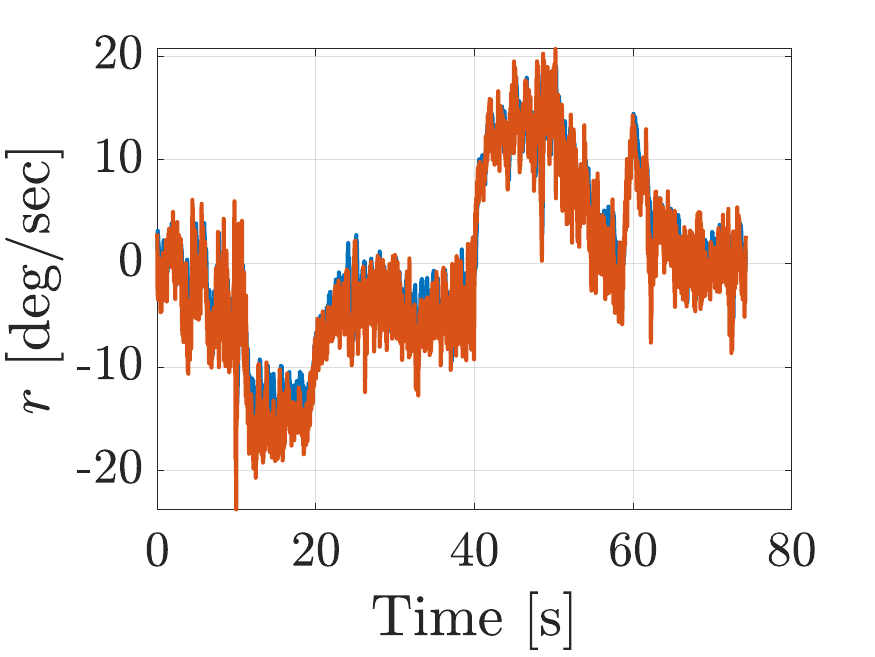}
	\caption{Non-adaptive NDI Controller - Nominal Trajectory}
	\label{fig:F2Angles}
\end{figure}
%----------------------------------------------

\emph{With \lonew compensation}:  The second flight of the experiment was similar to the first.
It began with an R/C takeoff, after which, the learning mode was engaged with PTIs for a lap at three different speeds, low, medium and high.  
After disabling the PTIs, pitch, yaw, and roll doublets were performed.  
This time, when autopilot was engaged, the controller became the NDI controller with the adaptive component.  
Once again there were no significant transients after the switch, and the vehicle began navigating as expected.
The tracking performance overall was good, as seen in Figure \ref{fig:F3Angles}, but a small offset can be seen in the pitch channel, e.g. $\theta$ and $q$.  
Similar results were obtained after engaging the autopilot in numerous other conditions, including various roll and heading angles, at a high altitude with slow speed, and a low altitude with high speed.
Satisfied with the performance, the pilot regained control with the PTIs active and flew several approaches to stall.  The autopilot was then engaged after an approach to stall.  As shown in Figure \ref{fig:F3Stall}, it recovered the vehicle as expected and navigated to its waypoints.  
Finally, the vehicle was returned to R/C mode for landing.

%----------------------- F3 Nominal Figures
\begin{figure}[htbp]
	\centering
	\includegraphics[trim={0.in 0in 0.4in 0.2in},clip,width = 0.32\linewidth]{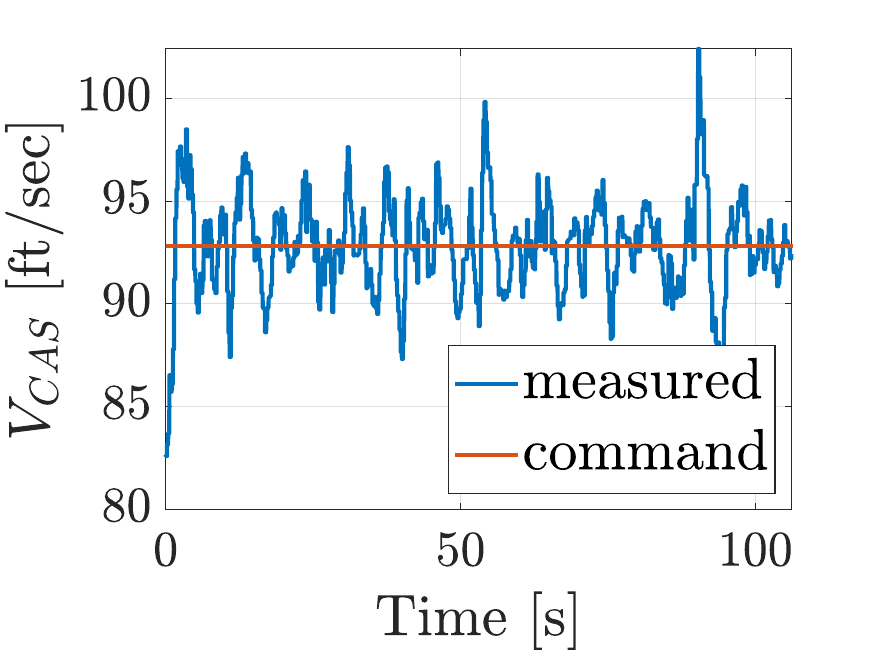}
	\includegraphics[trim={0.in 0in 0.4in 0.2in},clip,width = 0.32\linewidth]{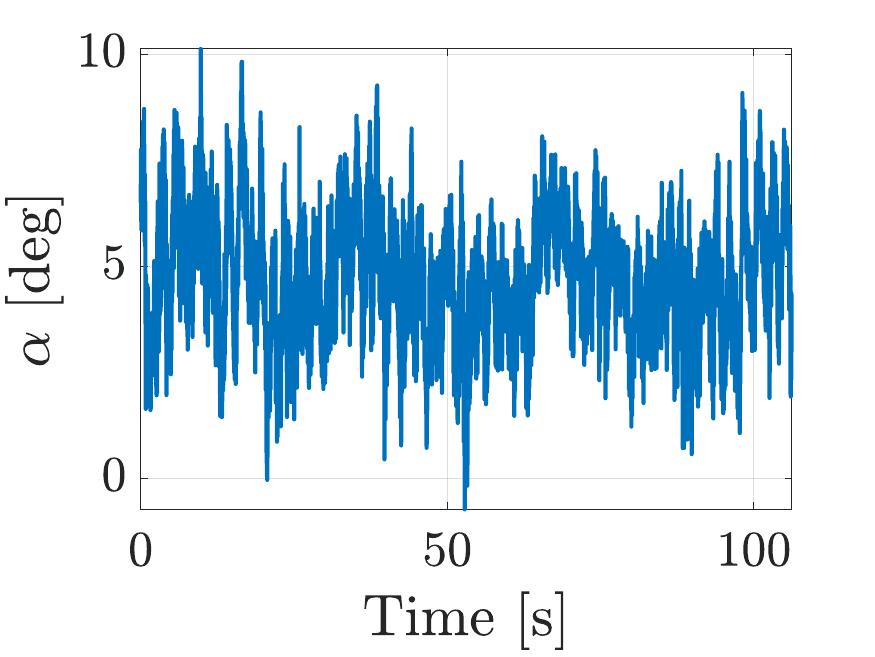}
	\includegraphics[trim={0.in 0in 0.4in 0.2in},clip,width = 0.32\linewidth]{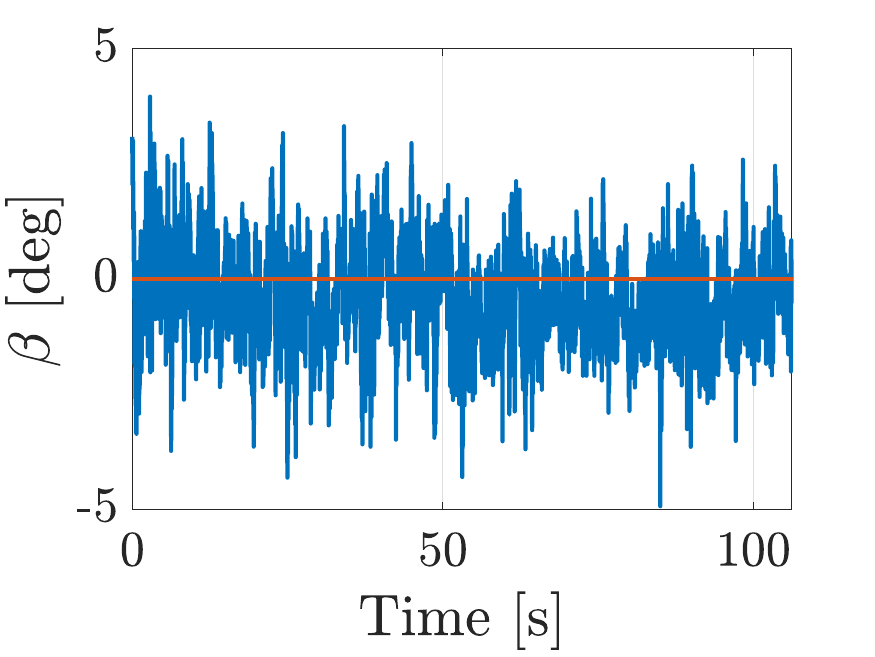}
	\includegraphics[trim={0.in 0in 0.4in 0.2in},clip,width = 0.32\linewidth]{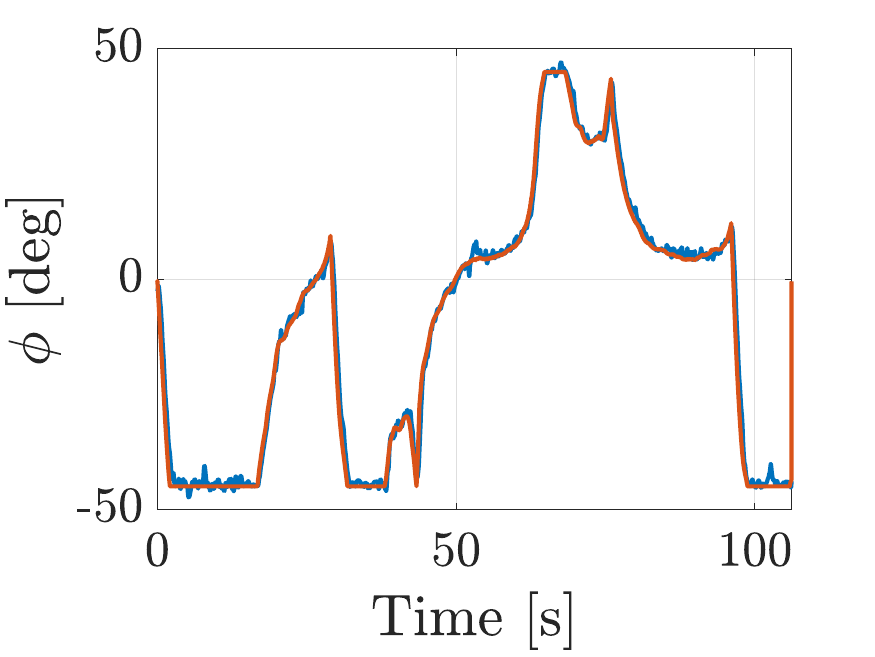}
	\includegraphics[trim={0.in 0in 0.4in 0.2in},clip,width = 0.32\linewidth]{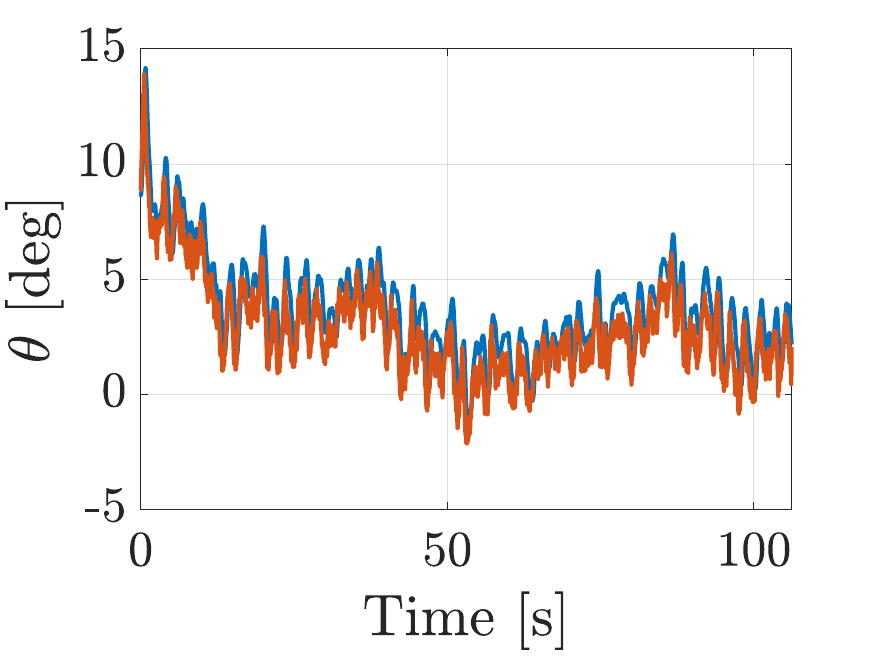}
	\includegraphics[trim={0.in 0in 0.4in 0.2in},clip,width = 0.32\linewidth]{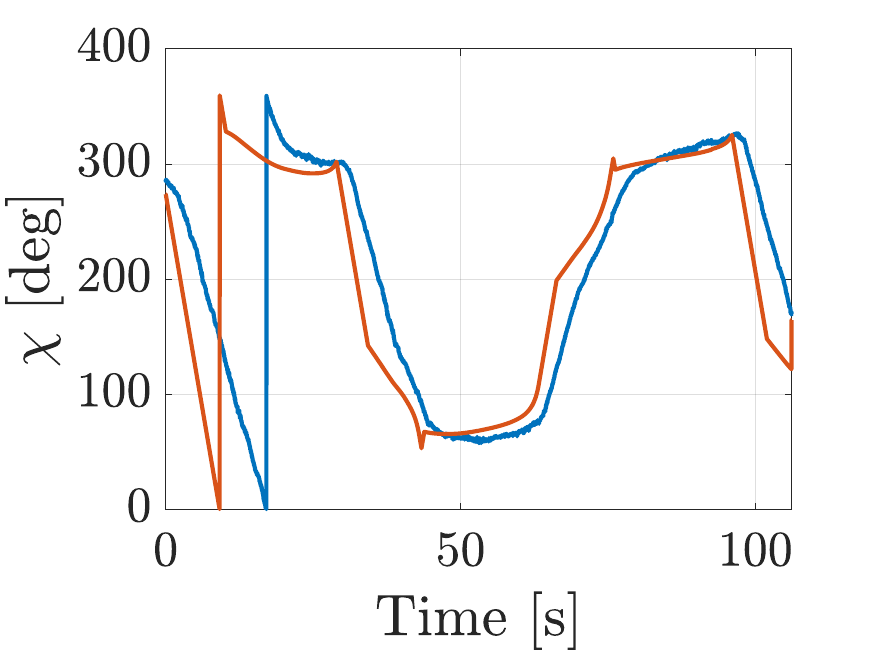}
	\includegraphics[trim={0.in 0in 0.4in 0.2in},clip,width = 0.32\linewidth]{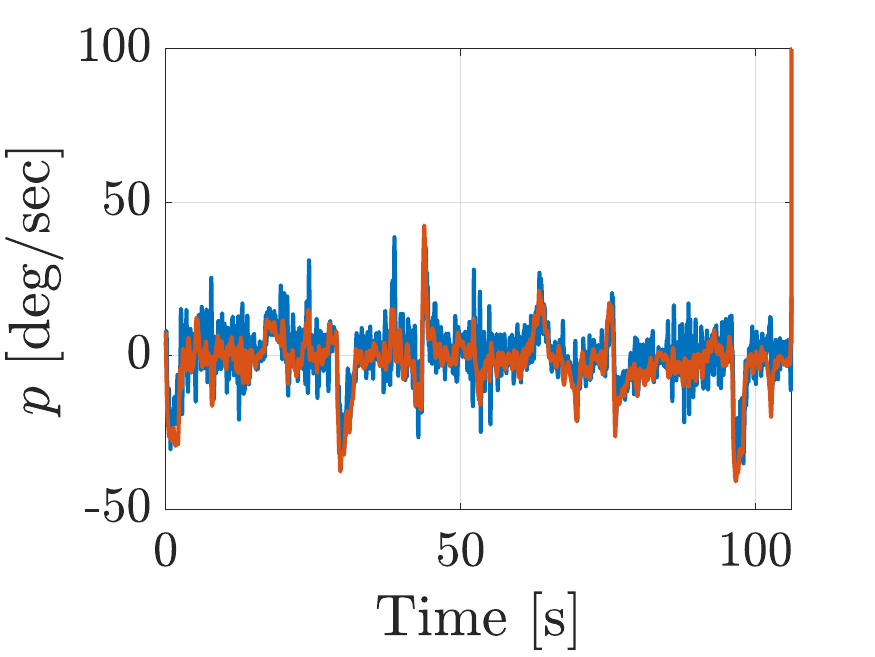}
	\includegraphics[trim={0.in 0in 0.4in 0.2in},clip,width = 0.32\linewidth]{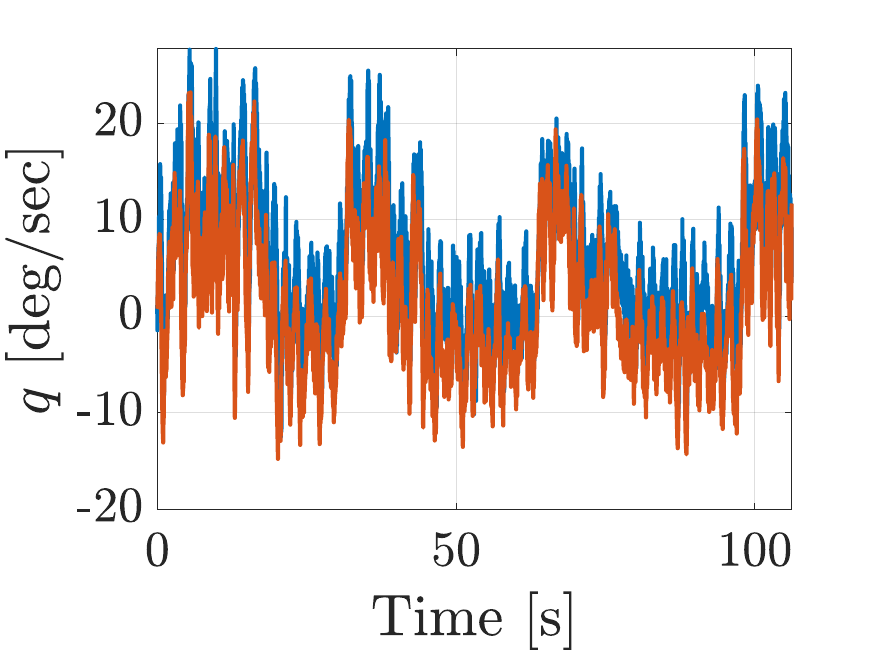}
	\includegraphics[trim={0.in 0in 0.4in 0.2in},clip,width = 0.32\linewidth]{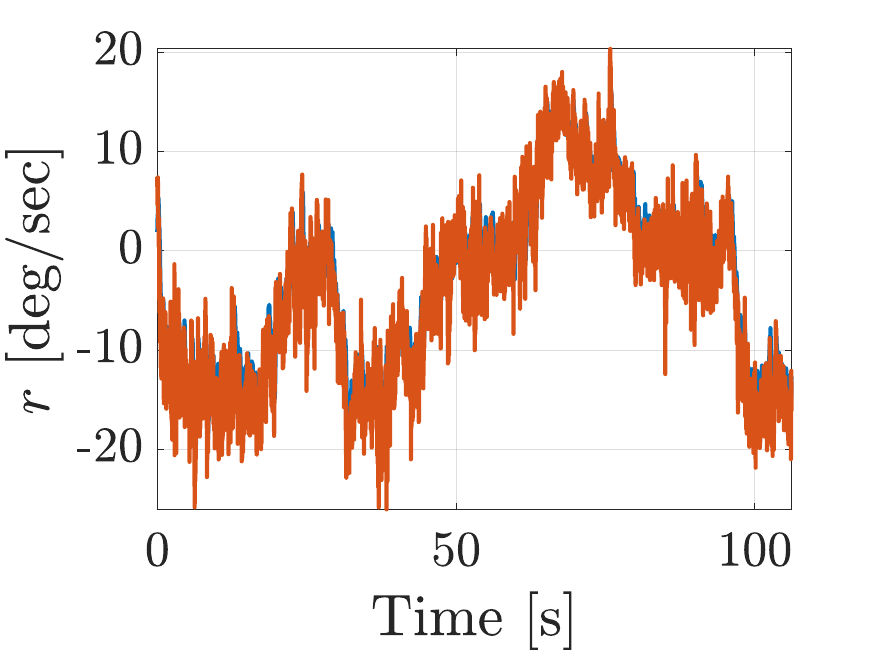}
	\caption{Adaptive NDI Controller - Nominal Trajectory}
	\label{fig:F3Angles}
\end{figure}
%----------------------------------------------

%--------------------------- F3 Stall Figures
\begin{figure}[htbp]
	\centering
	\includegraphics[trim={0.in 0in 0.4in 0.2in},clip,width = 0.32\linewidth]{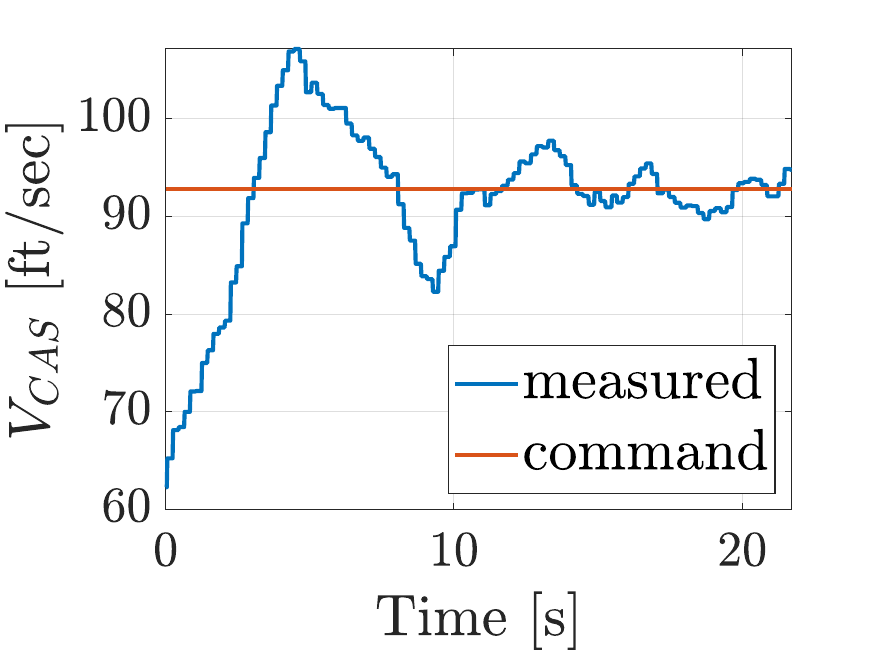}
	\includegraphics[trim={0.in 0in 0.4in 0.2in},clip,width = 0.32\linewidth]{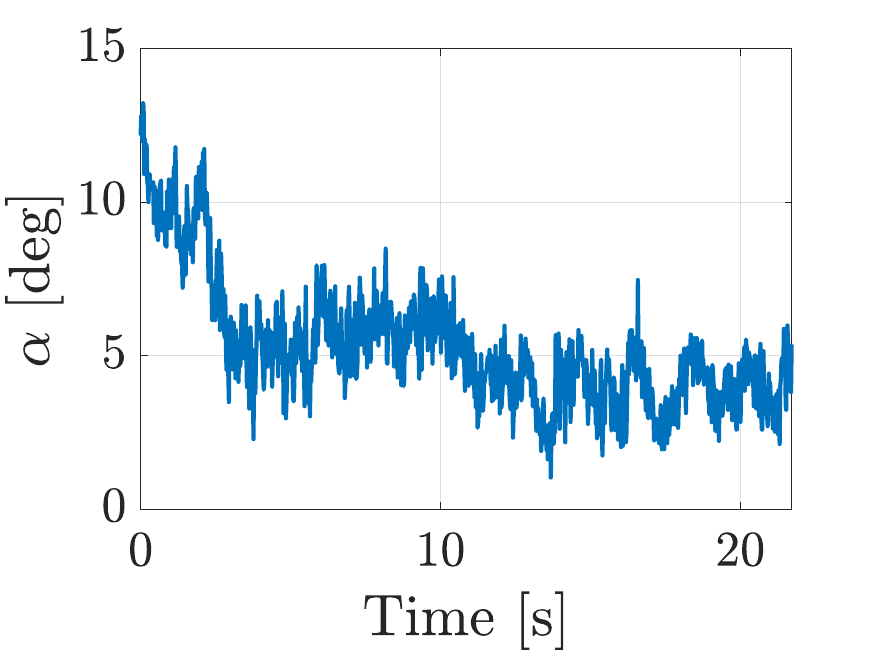}
	\includegraphics[trim={0.in 0in 0.4in 0.2in},clip,width = 0.32\linewidth]{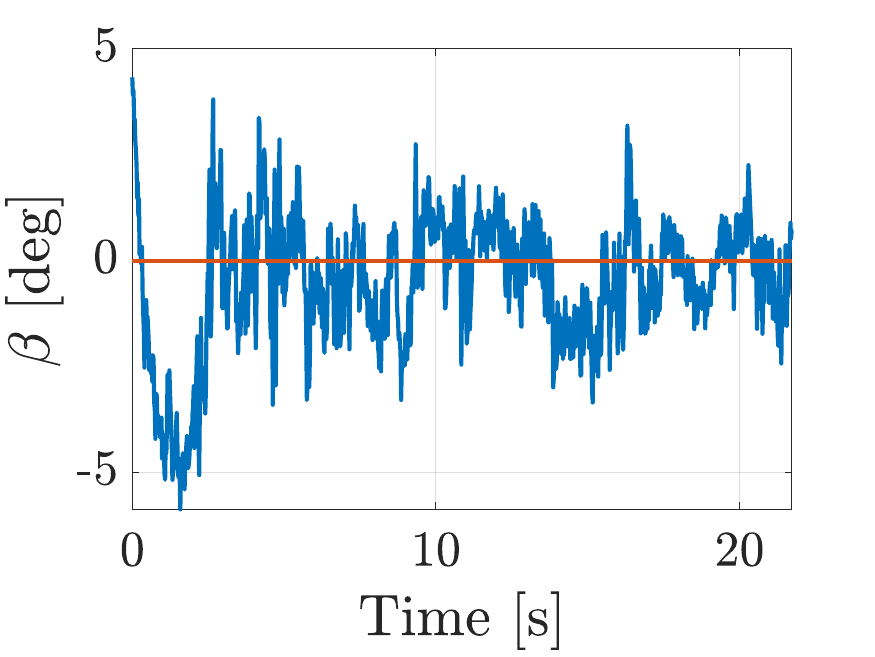}
	\includegraphics[trim={0.in 0in 0.4in 0.2in},clip,width = 0.32\linewidth]{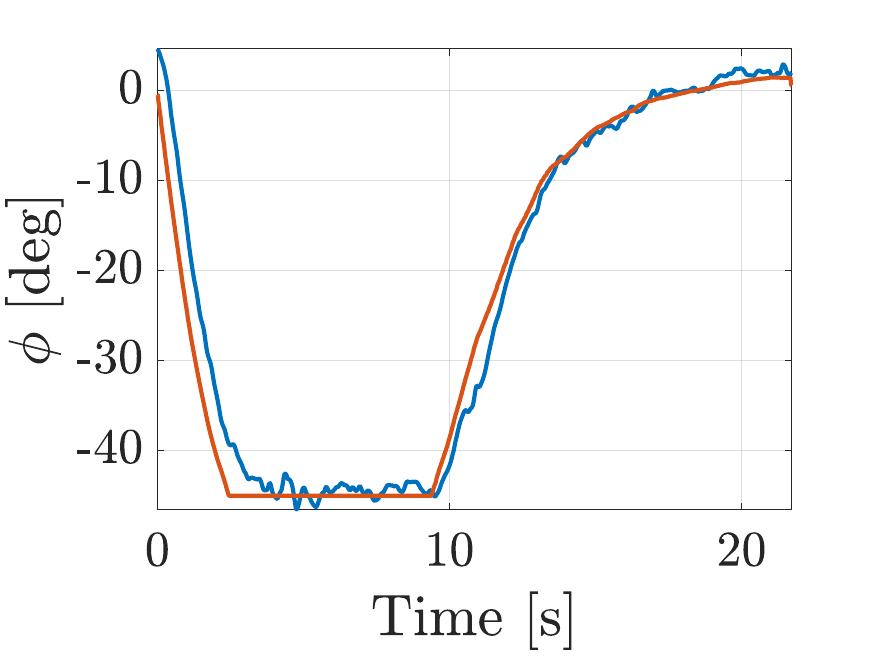}
	\includegraphics[trim={0.in 0in 0.4in 0.2in},clip,width = 0.32\linewidth]{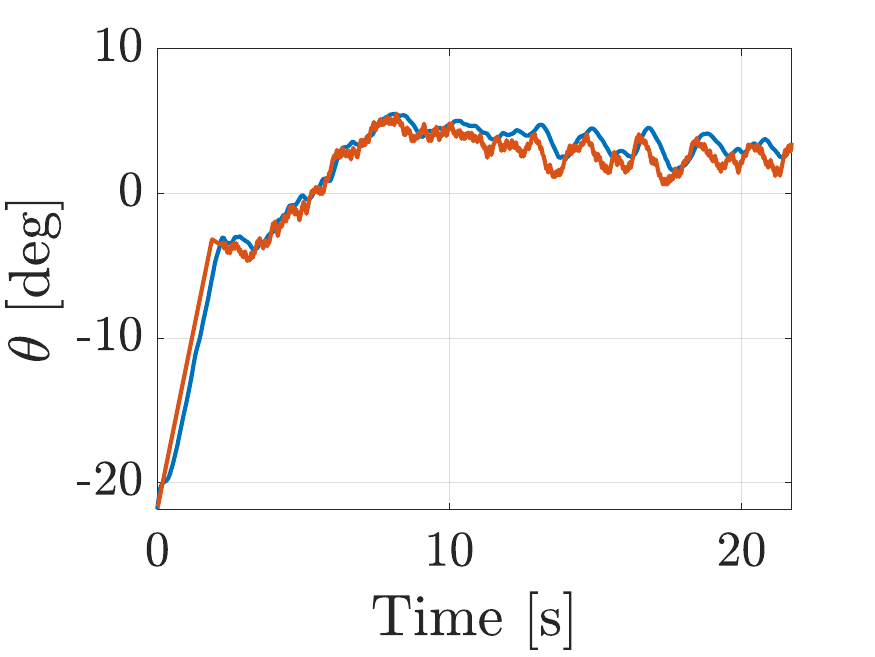}
	\includegraphics[trim={0.in 0in 0.4in 0.2in},clip,width = 0.32\linewidth]{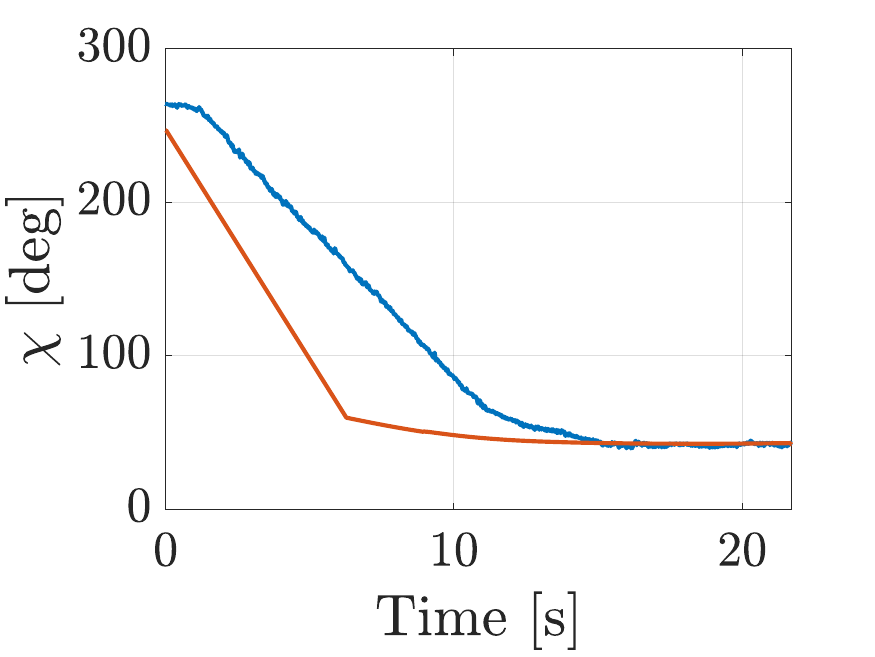}
	\includegraphics[trim={0.in 0in 0.4in 0.2in},clip,width = 0.32\linewidth]{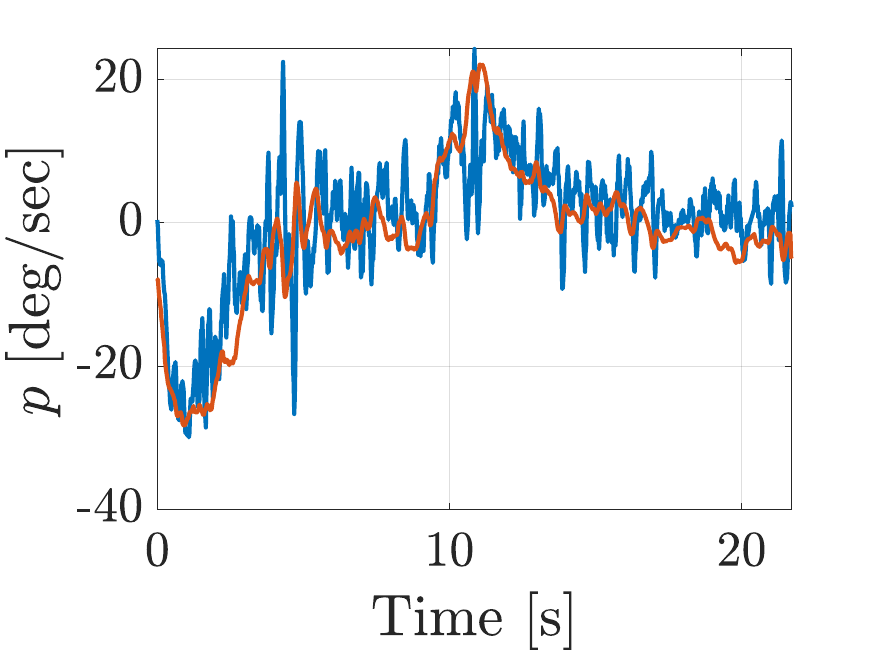}
	\includegraphics[trim={0.in 0in 0.4in 0.2in},clip,width = 0.32\linewidth]{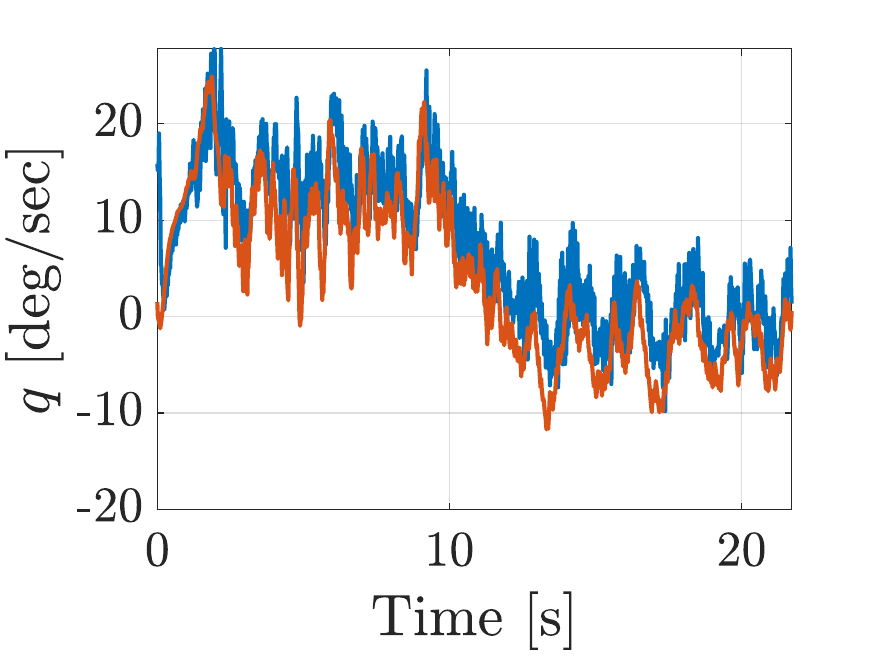}
	\includegraphics[trim={0.in 0in 0.4in 0.2in},clip,width = 0.32\linewidth]{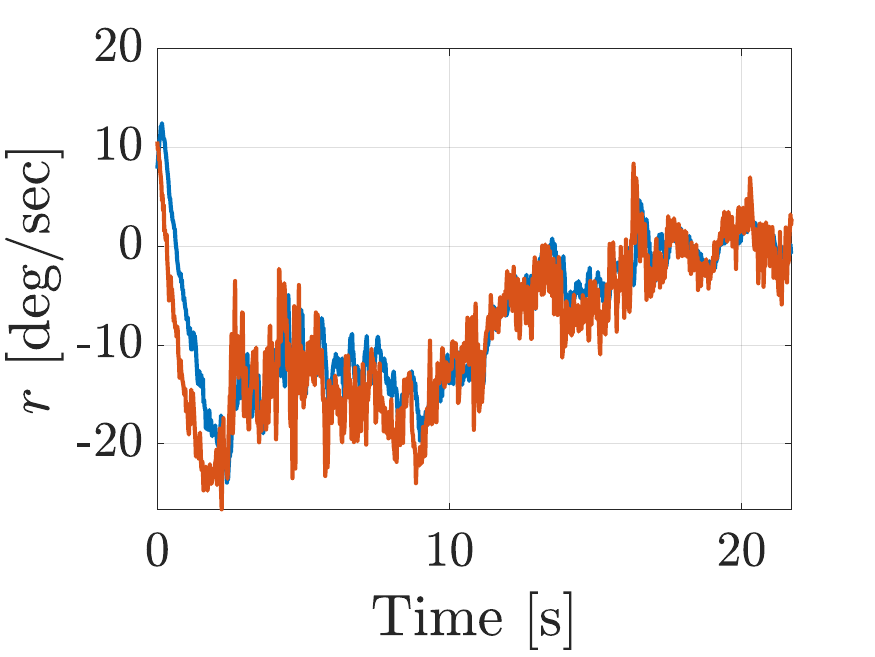}
	\caption{Adaptive NDI Controller - Approach to Stall}
	\label{fig:F3Stall}
\end{figure}
%----------------------------------------------

%----------------------------------------------------------------------------------------------------------------
\subsection{Control and learning performance in the presence of a poor initial model}
In the previous flights, the vehicle was manually flown through the learning phase while the PTIs were active.  
This was possible because the vehicle was stable and capable of being flown without the assistance of automatic control.
In the next series of flights, a hidden feedback loop was used to destabilize the vehicle, in the pitch axis for the first two flights and in the roll axis for the third flight.
This feedback loop was enabled at the same time as the learning mode and made the vehicle difficult to pilot manually.  
Thus in these flights, the learning and control occurred simultaneously.  
In addition, the actual (unstable) dynamics were significantly different from the (stable) initial model used by the controller, which is perfect for testing the capability of the L2F system.  

The first flight in this series of experiments used angle-of-attack feedback to the left elevator to produce roughly a -10\% static margin.  
The controller and modeling were not directly aware of this feedback and were blind to what the left elevator was doing; the modeling had no measurement of the left elevator's actions, and the controller could not send commands to the left elevator.  
The flight began with a piloted take off in the nominal configuration (no destabilizing feedback).  
Upon reaching the appropriate test altitude, the autonomous mode was engaged.
This activated the destabilizing feedback, the PTIs, and the autopilot.  
The vehicle was able to stabilize and navigate to the waypoints.
The autonomous mode was disengaged and reengaged several times, resulting in similar behavior.
Using large stick inputs to attempt to maintain straight and level flight, the R/C pilot was also able to fly the destabilized vehicle in pilot command pass-through mode, though moderate angle-of-attack and pitch-attitude excursions were observed.  
The vehicle was returned to its nominal configuration for approach and landing.  

The second flight was similar to the first but with a larger pitch-axis instability, about -16.4\% static margin 
\footnote{A video of one of the flight tests comparing the performance of the R/C pilot and the L2F autopilot for the destabilized vehicle is available at \url{https://youtu.be/y6O1mwzHdOE}}.
Once again, the nominal vehicle was taken up to a certain altitude before engaging the destabilizing feedback, PTIs, and autopilot.  
A couple of pitch oscillations occurred, but the vehicle recovered and navigated appropriately.  
Repeats produced similar results.
Unlike in the first flight, when the R/C pilot flew the destabilized aircraft, he was barely able to control it, even with full or almost-full stick deflections.  
Approach and landing were once again performed in the nominal configuration.

%-------------------------------- F11 Figures
\begin{figure}[htbp]
	\centering
	\includegraphics[trim={0.in 0in 0.4in 0.2in},clip,width = 0.32\linewidth]{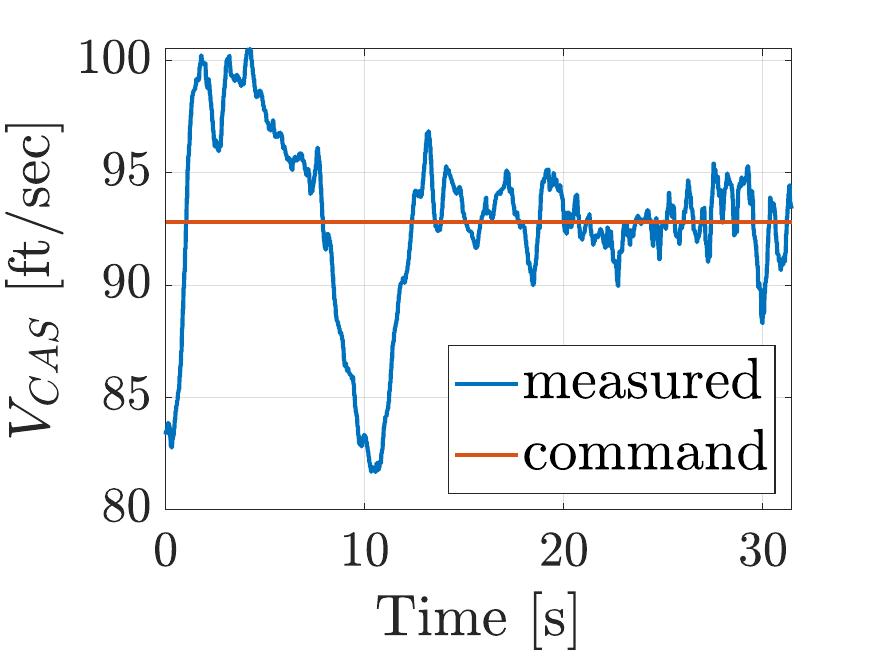}
	\includegraphics[trim={0.in 0in 0.4in 0.2in},clip,width = 0.32\linewidth]{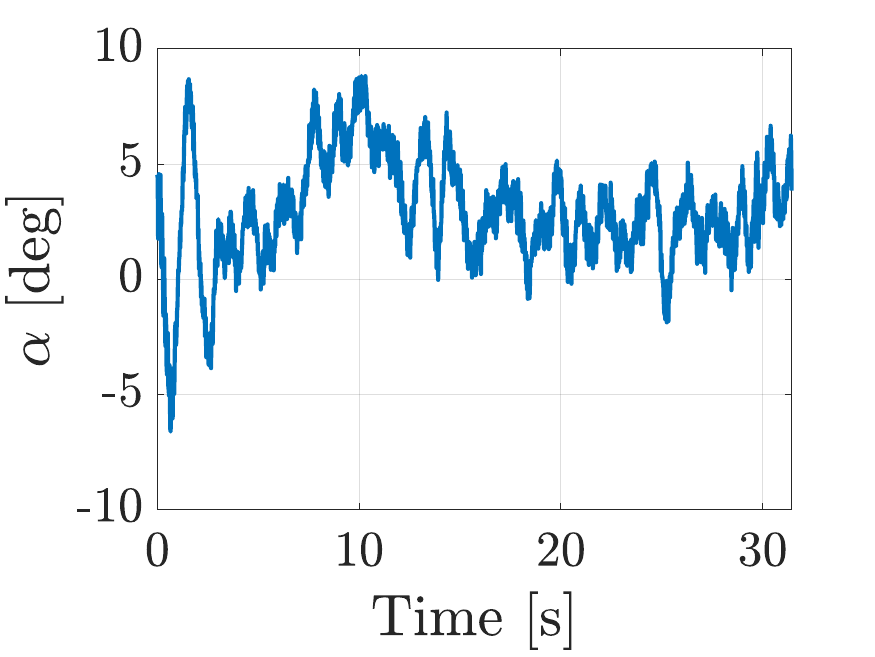}
	\includegraphics[trim={0.in 0in 0.4in 0.2in},clip,width = 0.32\linewidth]{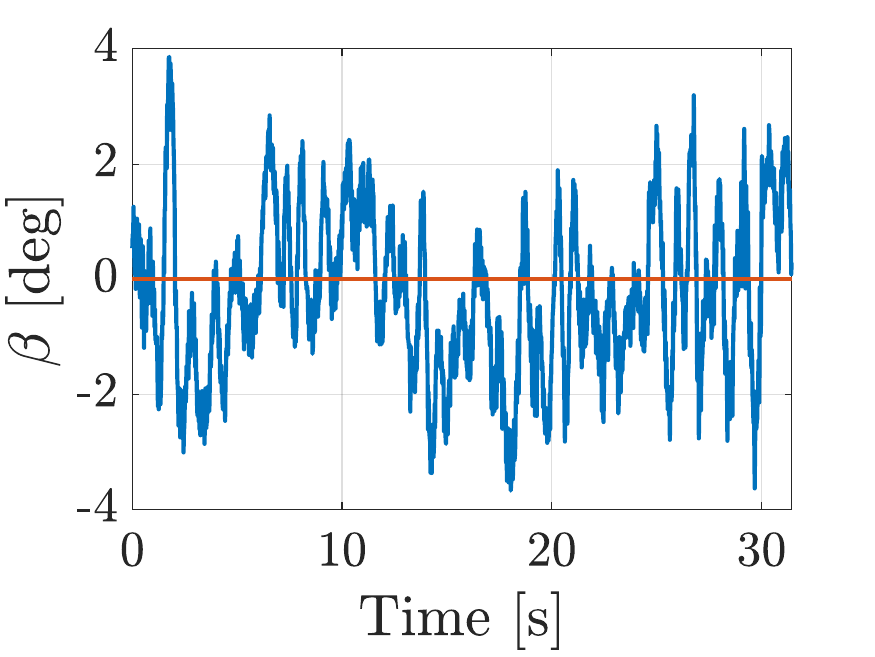}
	\includegraphics[trim={0.in 0in 0.4in 0.2in},clip,width = 0.32\linewidth]{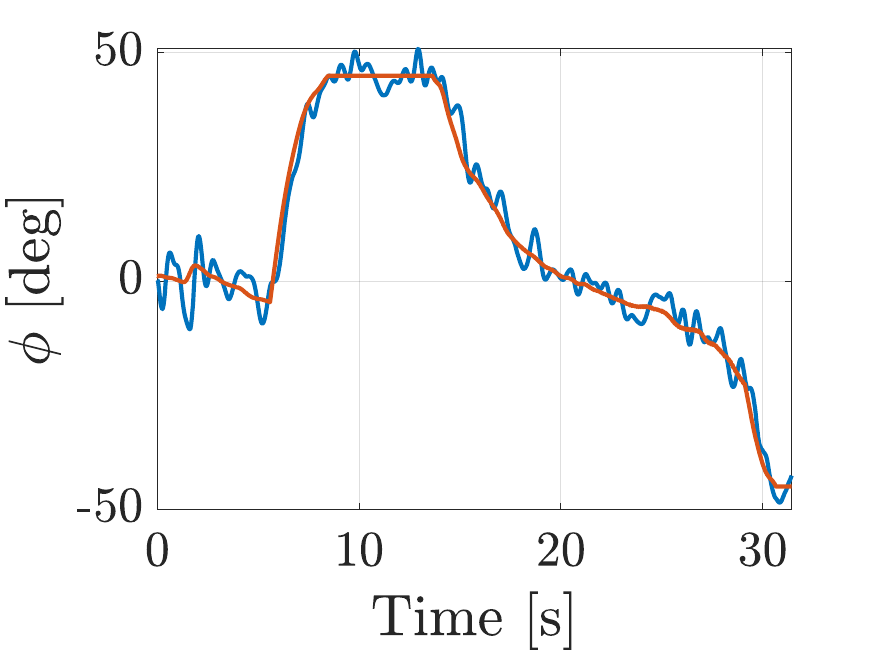}
	\includegraphics[trim={0.in 0in 0.4in 0.2in},clip,width = 0.32\linewidth]{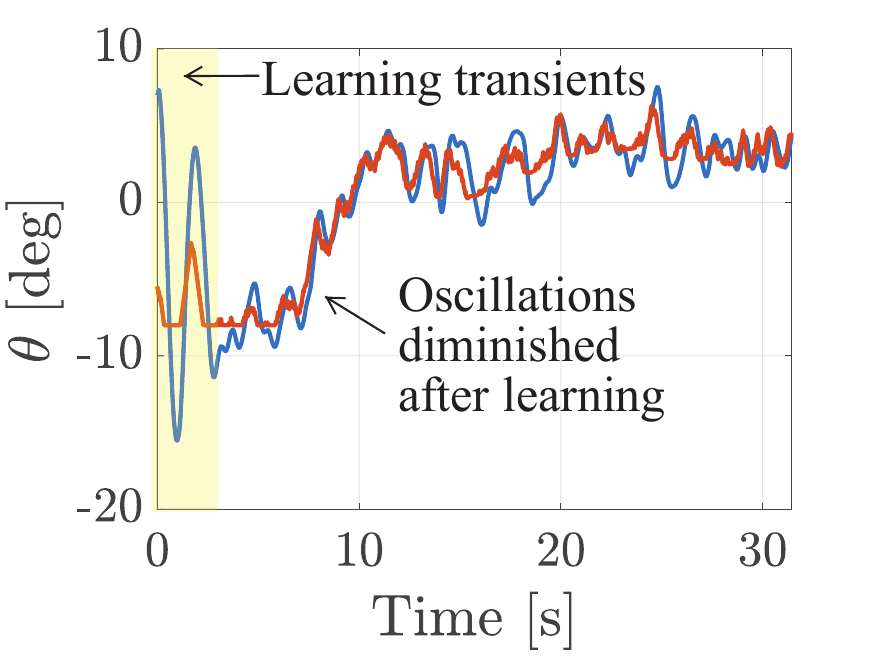}
	\includegraphics[trim={0.in 0in 0.4in 0.2in},clip,width = 0.32\linewidth]{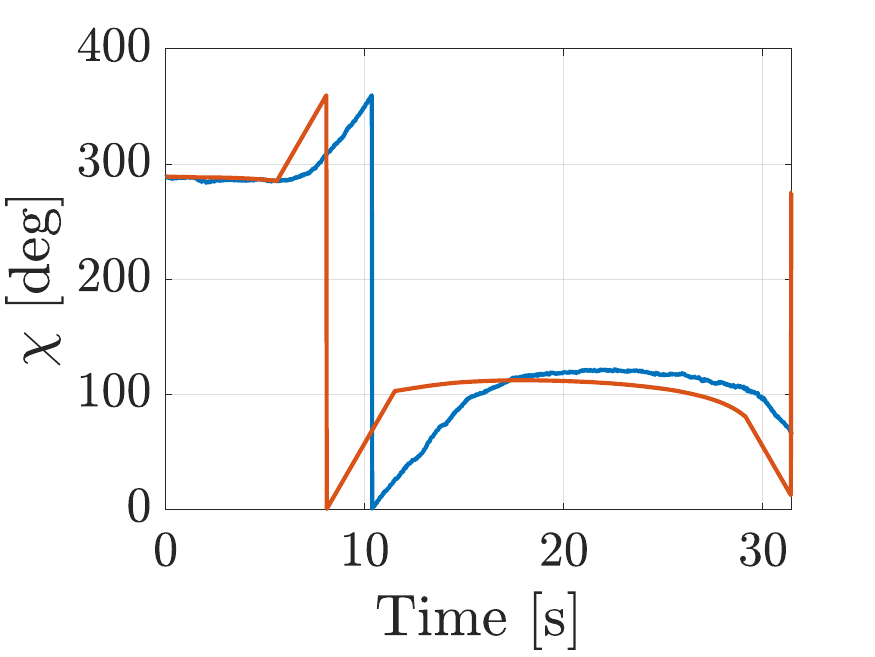}
	\includegraphics[trim={0.in 0in 0.4in 0.2in},clip,width = 0.32\linewidth]{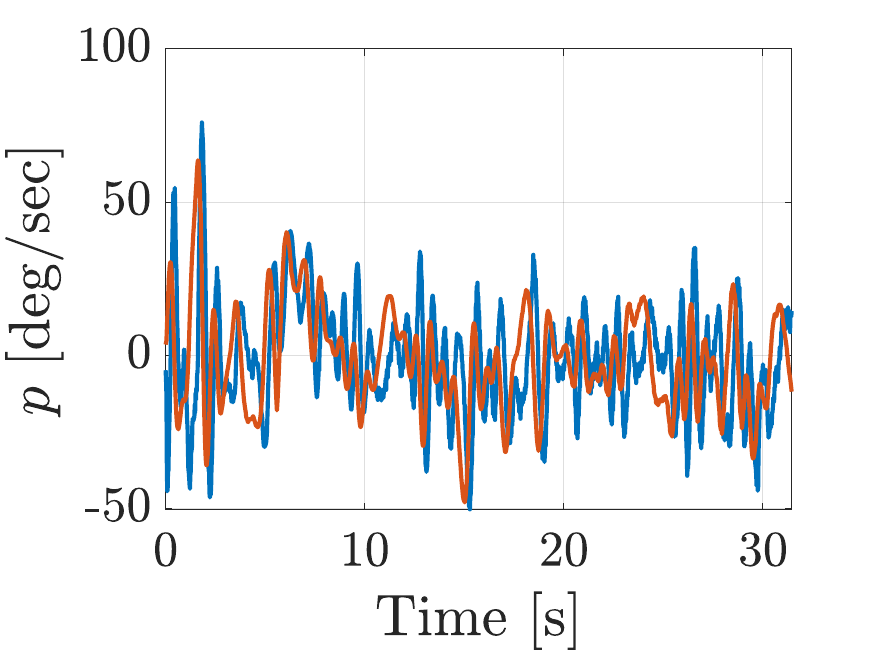}
	\includegraphics[trim={0.in 0in 0.4in 0.2in},clip,width = 0.32\linewidth]{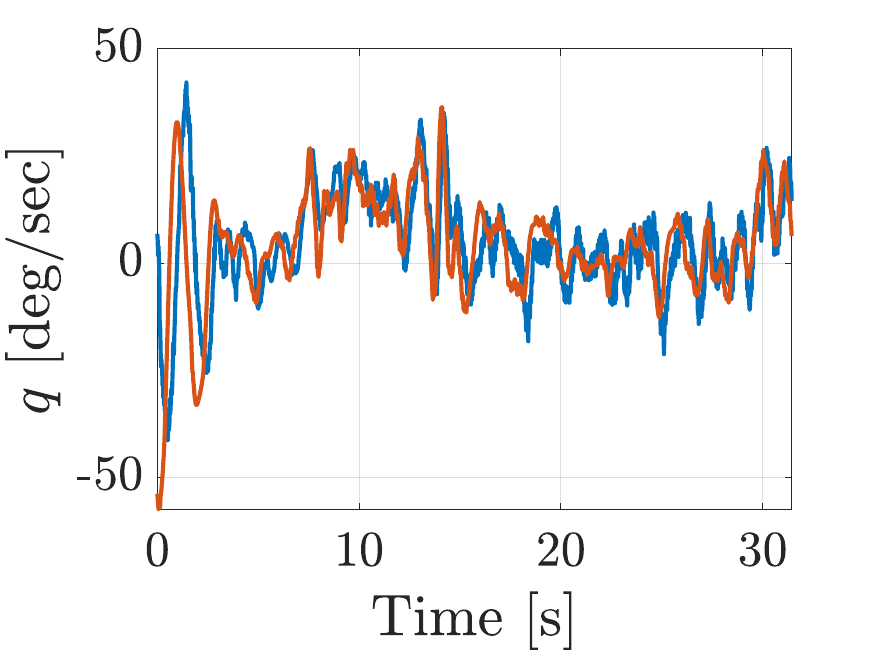}
	\includegraphics[trim={0.in 0in 0.4in 0.2in},clip,width = 0.32\linewidth]{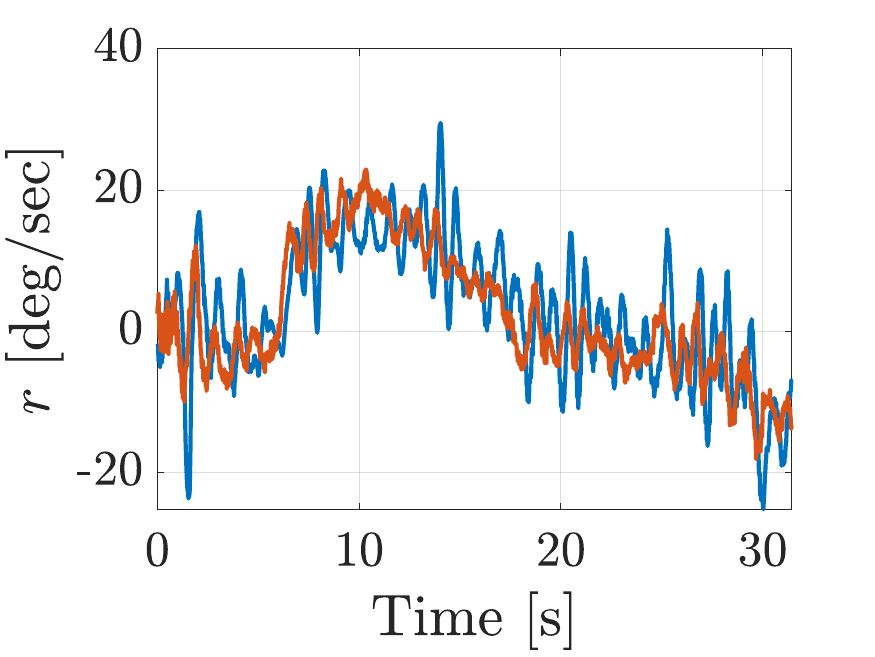}
	\caption{Adaptive NDI Controller - Vehicle Pitch Axis Destabilization}
	\label{fig:F11Angles}
\end{figure}
%----------------------------------------------

Figure \ref{fig:F11Angles} shows plots of the autonomous system's tracking performance from the second flight. 
Some transitory errors occurred at the beginning of the flight.
However, after a few seconds, the vehicle settled and began tracking its commands with improved accuracy.  
In Figure \ref{fig:F11uAd}, the adaptive command is plotted along with the pitch angle.  
After the learning transient, the vehicle was able to update its knowledge of the vehicle and the magnitude of the adaptive command reduced.  
Despite the improved system knowledge, uncertainties and disturbances still exist, such as atmospheric disturbances, the PTIs, etc.
Because of this, the adaptive command did not necessarily drop to zero after learning.
%---------------------------- F11 Adaptive Input vs Theta
\begin{figure}[htbp]
	\centering
	\includegraphics[trim={0in 0in 0in 0.1in},clip,width = 0.40\linewidth]{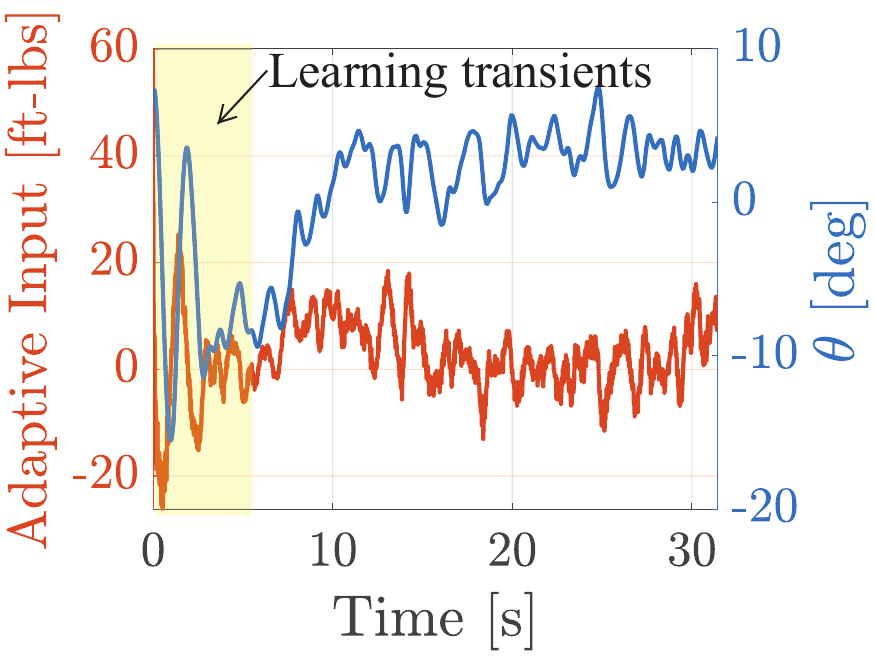}
	\caption{Adaptive Input vs Pitch Angle for Pitch Axis Destabilization Case}
	\label{fig:F11uAd}
\end{figure}
%----------------------------------------------

The pilot's flight data (red) and that of the autonomous system (blue) are overlaid in Fig. \ref{fig:F11Pilot}.
When the vehicle was controlled by the pilot, despite large inputs ($\delta_e$) executed by the pilot, large pitch oscillations ($\theta$) still occurred.
In comparison, in the autonomous control mode, the real-time modeling was able to quickly determine that the vehicle was unstable (positive $C_{M_\alpha}$ estimate), and the control law modified the control gain ($K_\alpha$) to compensate.
$K_\alpha$ is used in the construction of the portion of the estimated moment $\hat{M}$ that depends on $\alpha$.
It is important to note that there will always be an inherent delay within the modeling results because sufficient informative data must be collected before a good model can be obtained.  
While real-time modeling was gathering the data, the baseline controller was still operating based on a stable aircraft model (the initial guess) for the first few seconds.
This is where {adaptive compensation is critical for maintaining the stability of the vehicle}.  
Note that due to use of the initial poor model, the actual uncertainties (i.e., the discrepancy between this model and the actual dynamics) most probably included unmatched and/or nonlinear uncertainties that are not considered in Section \ref{sec::L1AC}.
As a result, the tracking performance was not as good as claimed in Theorem \ref{thm:ebnd}.  
For instance, we can see a large pitch excursion in Fig. \ref{fig:F11Pilot} at the beginning where the controller was struggling to achieve good tracking performance.  
Nevertheless, thanks to the adaptive compensation, the vehicle was able to maintain its stability, which is a major safety concern for an aerial vehicle.
Once the system collected sufficient data to learn and update the model, the performance improved, resulting in relatively smooth, level flight.  

%---------------------------- F11 Pilot vs L2F
\begin{figure}[htbp]
	\centering
	\includegraphics[width = 1.0\linewidth]{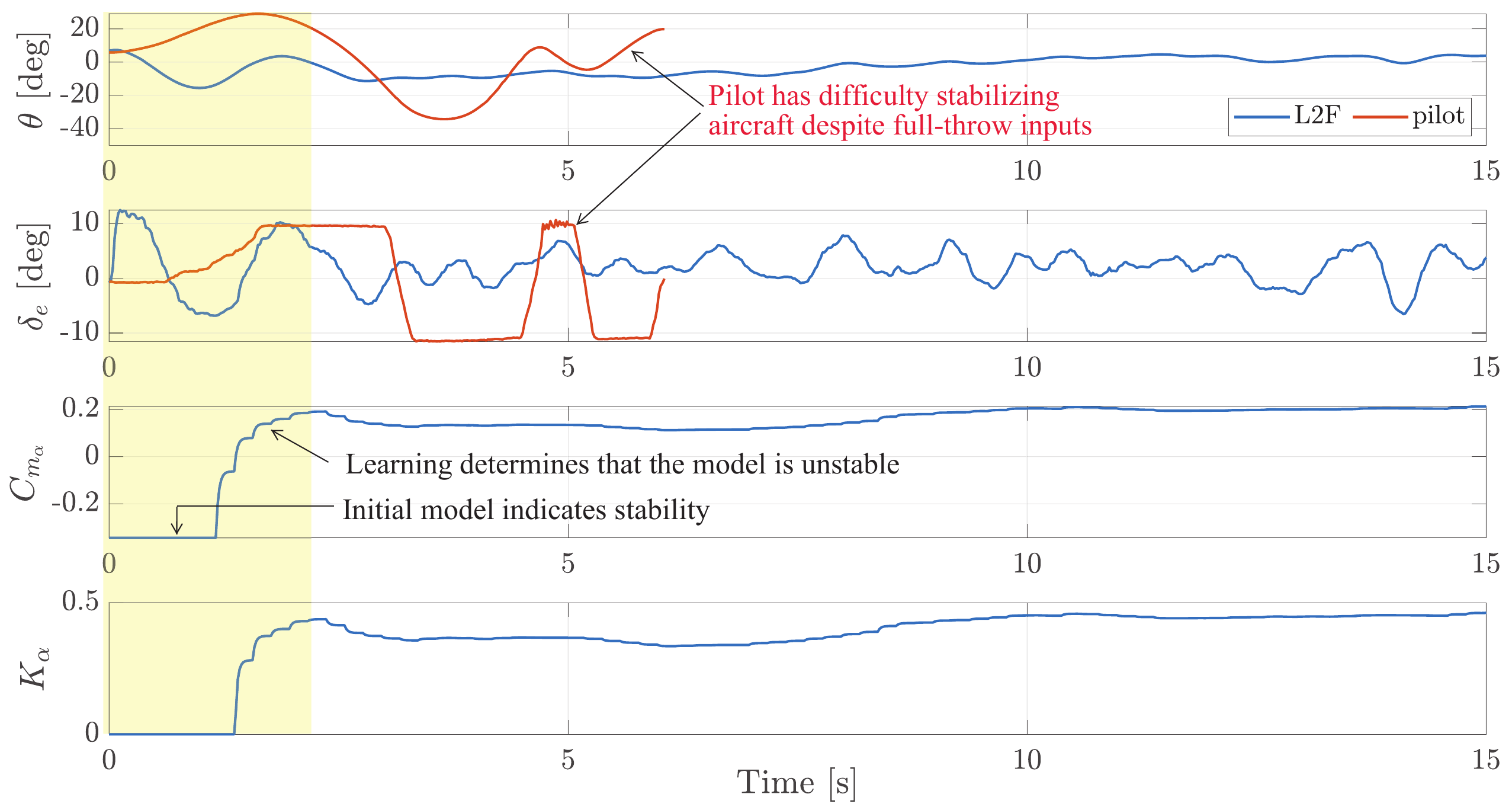}  
	\caption{Overlay of Pilot Data with Autonomous Data}
	\label{fig:F11Pilot}
\end{figure}
%----------------------------------------------

The third flight in this series used roll rate feedback to the inboard flaps in order to destabilize the roll axis. 
Initially, the target time to double was $0.5$ seconds, but the feedback gain used was only able to render the vehicle neutrally stable.  
While the target level of instability was not achieved, it was sufficient to degrade the innate stability of the vehicle.  
Similar to the last two flights, the pilot first took off in the nominal configuration and then enabled the destabilizing feedback and PTIs and activated the autopilot mode once the vehicle reached the test altitude.  
Two engagements of the destabilizing feedback produced similar results with a few initial roll oscillations before navigating as expected.
The R/C pilot also made two attempts to fly the destabilized vehicle.  
In the first, the pilot was able to stabilize the vehicle, but with some large roll oscillations.
On the second, the pilot maintained better control of the vehicle but noted that it was difficult to fly.
After reverting to the stable configuration, the pilot landed the vehicle.  

Tracking data for this third flight can be seen in Fig. \ref{fig:F10Angles}.
Like the first two flights, there were large tracking errors at the beginning (though this time in the roll axis),  but the L2F system quickly learned the new dynamics and the oscillations died down momentarily as the control laws were updated with more accurate models.

%-------------------------------- F10 Figures
\begin{figure}[hb]
	\centering
	\includegraphics[trim={0.in 0in 0.4in 0.2in},clip,width = 0.32\linewidth]{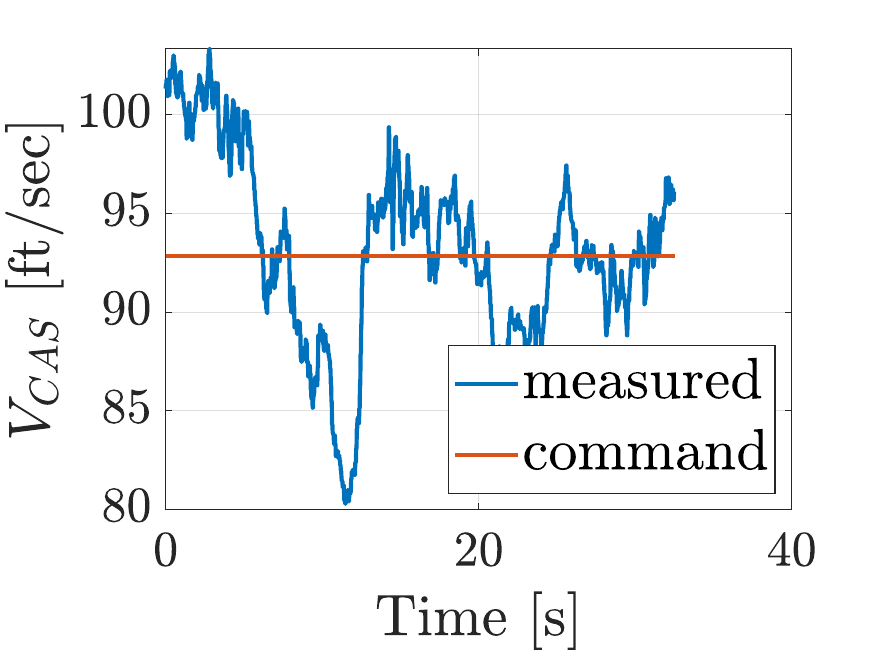}
	\includegraphics[trim={0.in 0in 0.4in 0.2in},clip,width = 0.32\linewidth]{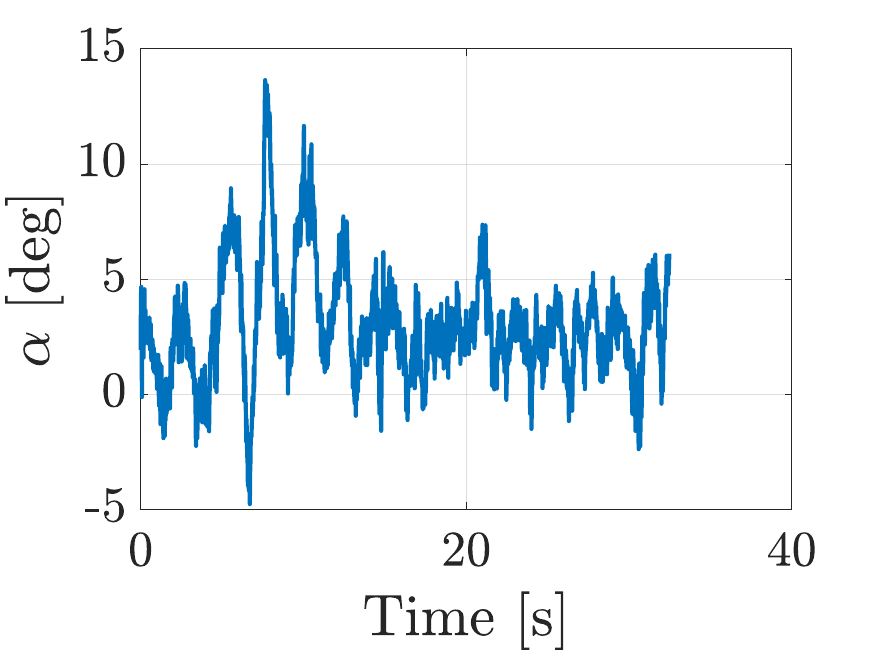}
	\includegraphics[trim={0.in 0in 0.4in 0.2in},clip,width = 0.32\linewidth]{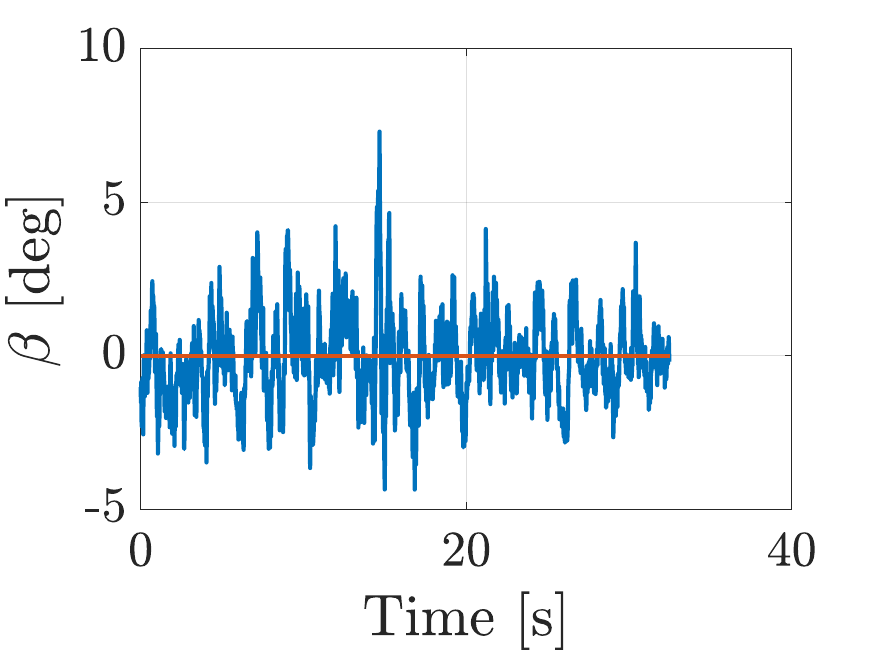}
	\includegraphics[trim={0.in 0in 0.4in 0.2in},clip,width = 0.32\linewidth]{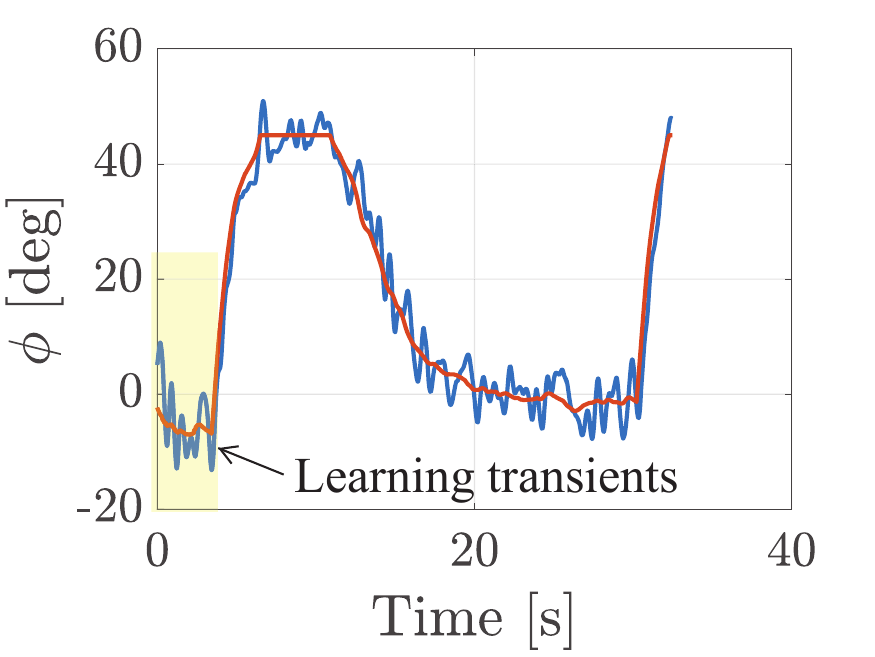}
	\includegraphics[trim={0.in 0in 0.4in 0.2in},clip,width = 0.32\linewidth]{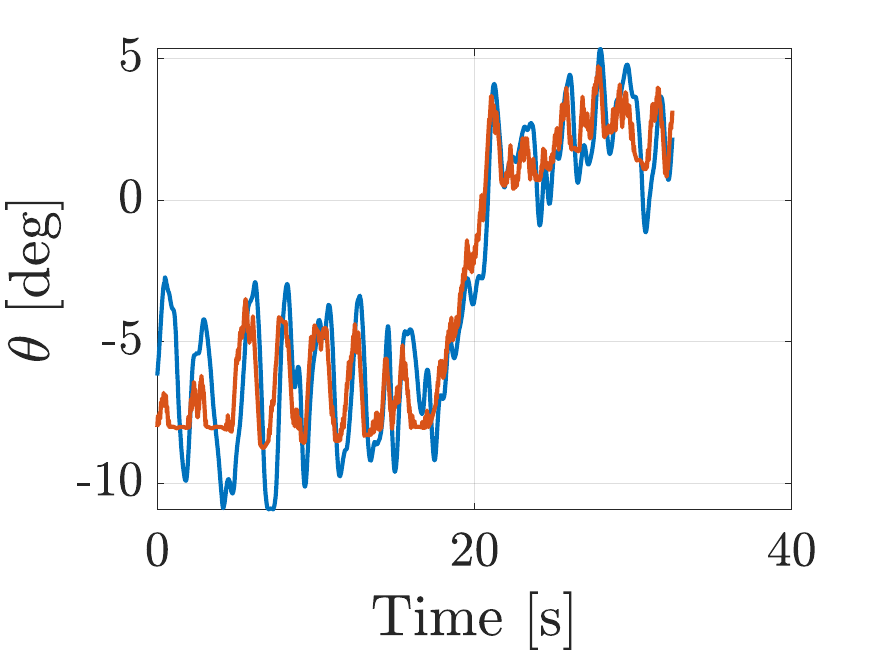}
	\includegraphics[trim={0.in 0in 0.4in 0.2in},clip,width = 0.32\linewidth]{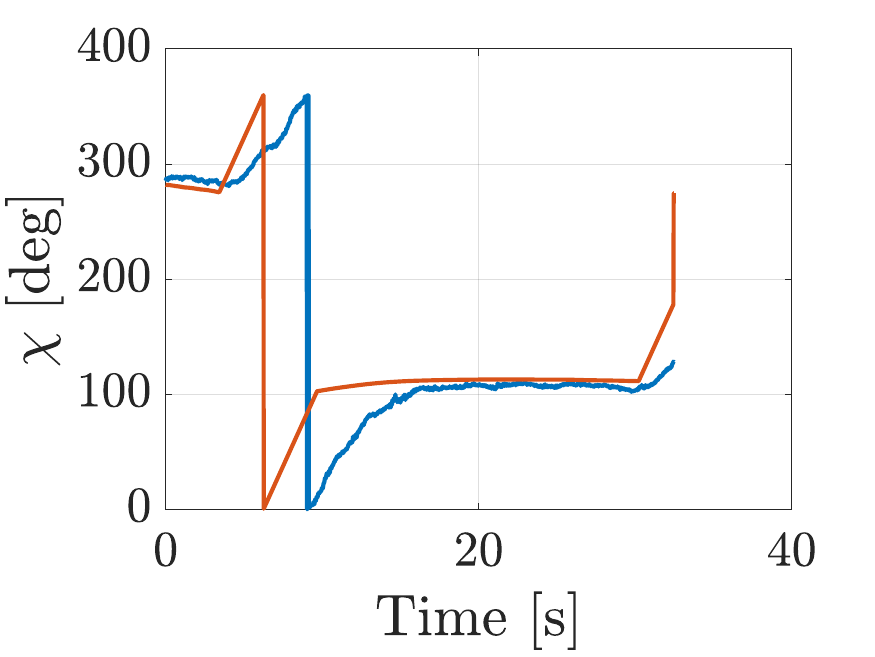}
	\includegraphics[trim={0.in 0in 0.4in 0.2in},clip,width = 0.32\linewidth]{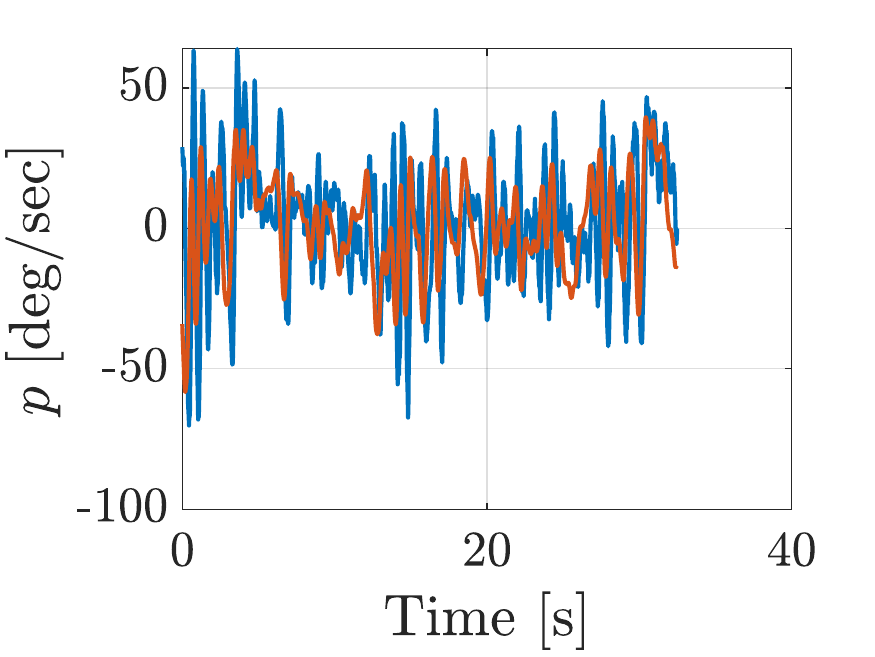}
	\includegraphics[trim={0.in 0in 0.4in 0.2in},clip,width = 0.32\linewidth]{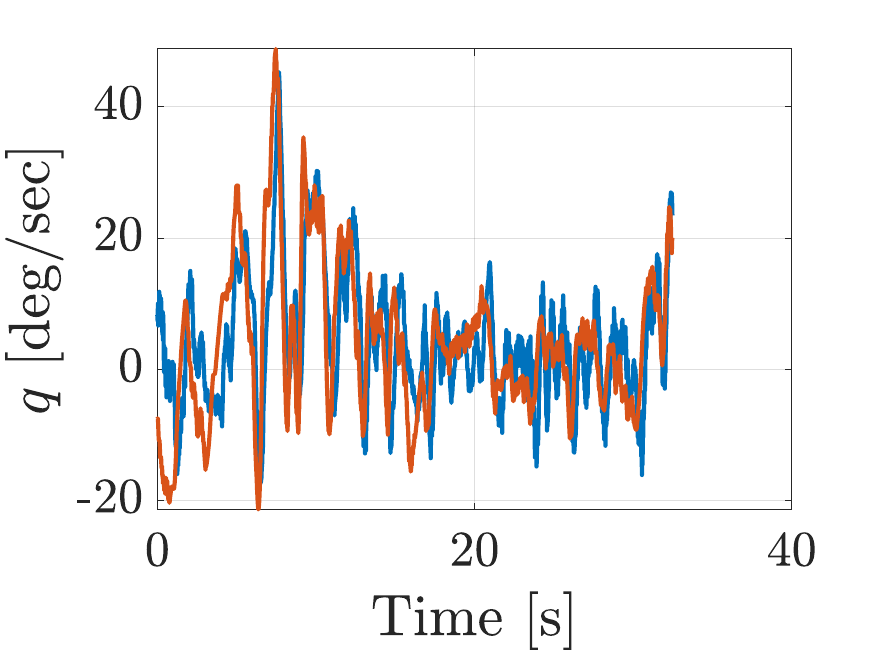}
	\includegraphics[trim={0.in 0in 0.4in 0.2in},clip,width = 0.32\linewidth]{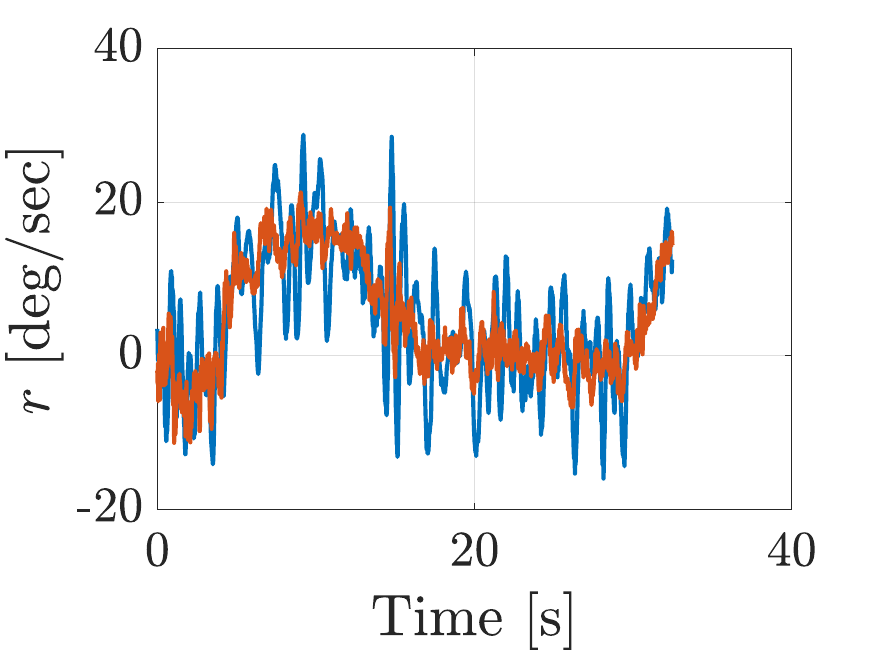}
	\caption{Adaptive NDI Controller - Vehicle Roll Axis Destabilization}
	\label{fig:F10Angles}
\end{figure}
%----------------------------------------------

\section{Conclusions}
This paper presents an \lonew adaptive control (\loneAC) based scheme to ensure safe learning of aerial vehicle dynamics within the Learn-to-Fly (L2F) framework, which is envisioned to replace the traditional iterative development paradigm for aerial vehicles. At the core of the proposed scheme is an \loneAC ~architecture that provides stability and transient performance guarantees for a switched linear reference system subject to unknown time-varying parameters and disturbances. The paper also includes the controller design for the L2F system and how the \loneAC ~law is designed and integrated into the L2F system. The proposed scheme was validated by flight tests on an unmanned aerial vehicle. 

From the flight tests, we can see that the benefit of using \lonew adaptive control within the L2F paradigm is that we can leverage the robustness of the \lonew controller to maintain the stability of the vehicle during the learning of its dynamics.  
At least initially, the uncertainties within the vehicle model are large.  
Whatever controller is used, it must be able to tolerate these large uncertainties to provide enough time for the dynamics to be learned.   
When a poor initial model is used, large tracking errors can occur even with \lonew compensation (see Figure \ref{fig:F11Angles}).
Without that robust compensation (which was not tested due to its high-risk nature), the vehicle's loss of stability could have occurred.
If we chose a fixed robust controller, then we would not take full advantage of the additional information that is being learned throughout the flight.  
By using an \lonew adaptive controller, we have the robustness to enable learning, and by updating (the desired dynamics of) the \lonew adaptive controller, we can make use of the learned model to improve the control performance.

Our future work includes extension of the \loneAC ~architecture for switched systems subject to more general uncertainties (e.g., state-dependent and unmatched \cite{zhao2020RALPV} uncertainties).

\section*{Acknowledgments}
This work is funded in part by the Air Force Office of Scientific Research (AFOSR) under the grant FA9550-18-1-0269, and in part by the Langley Research Center of the National Aeronautics and Space Administration (NASA) under the grant 80NSSC17M0051. 

\bibliography{L2FL1}
\end{document}